\documentclass[a4paper,reqno,onecolumn,11pt,accepted=2024-07-15]{quantumarticle}
\pdfoutput=1

\usepackage[utf8]{inputenc}
\usepackage{amsmath,bbm,physics,hyperref}
\hypersetup{citecolor=blue}
\usepackage{cleveref}
\usepackage[style=numeric,maxbibnames=99,giveninits=true]{biblatex}
\DefineBibliographyExtras{UKenglish}{}
\usepackage{import}
\usepackage{packages}
\usepackage{definitions}

\title[Trotter and Zeno Product Formulas]{On Strong Bounds for Trotter and Zeno Product Formulas with Bosonic Applications}
\author{Tim M\"obus}
\affiliation{Department of Mathematics, Technical University of Munich,  M\"unchen, Germany}
\affiliation{Munich Center for Quantum Science and Technology (MCQST),  M\"unchen, Germany}
\orcid{0000-0003-0471-2745}
\email{moebustim@gmail.com}

\begin{document}
\maketitle 
\begin{abstract}
	The Trotter product formula and the quantum Zeno effect are both indispensable tools for constructing time-evolutions using experimentally feasible building blocks. In this work, we discuss assumptions under which quantitative bounds can be proven in the strong operator topology on Banach spaces and provide natural bosonic examples. Specially, we assume the existence of a continuously embedded Banach space, which relatively bounds the involved generators and creates an invariant subspace of the limiting semigroup with a stable restriction. The slightly stronger assumption of admissible subspaces is well-recognized in the realm of hyperbolic evolution systems (time-dependent semigroups), to which the results are extended. By assuming access to a hierarchy of continuously embedded Banach spaces, Suzuki-higher-order bounds can be demonstrated. In bosonic applications, these embedded Banach spaces naturally arise through the number operator, leading to a diverse set of examples encompassing notable instances such as the Ornstein-Uhlenbeck semigroup and multi-photon driven dissipation used in bosonic error correction.
\end{abstract}
\tableofcontents

\newpage
\section{\bfseries Introduction}\label{sec:introduction}
	The experimental motivation for operator product formulas, particularly the (symmetrized) Trotter product formula and the quantum Zeno effect, is the approximation of a time evolution by simpler building blocks. An instance from quantum dynamics is the approximation of a global Hamiltonian evolution by simpler locally defined Hamiltonians \cite{Nielsen.2012}. On an abstract structure-wise level, we establish convergence rates for differences like
	\begin{equation*}
		\|\prod_{j=1}^nF_j(t,s) x - T(t,s)x\|\leq\,?,
	\end{equation*}
	where ${F_j(t,s)}_{j=1}^n$ is a sequence of contractions, $T(t,s)$ a contractive strongly continuous ($C_0$-) evolution system (time-dependent $C_0$-semigroup), and $x$ an element of the underlying Banach space. In the present paper, the building blocks ${F_j(t,s)}_{j=1}^n$ exhibit the structure given by the Trotter product formula or the quantum Zeno effect. For example, the former is given by $T(t,0)=e^{t(A+B)}$ and  $F_j=e^{t\frac{A}{n}}e^{t\frac{B}{n}}$ for two generators $A,B$ of contractive $C_0$-semigroups, for which \citeauthor{Trotter.1959} proved that
	\begin{equation}\label{eq:trotter}
		\Bigl\|\left(e^{t\frac{A}{n}}e^{t\frac{B}{n}}\right)^nx-e^{t(A+B)}x\Bigr\|\rightarrow0 \quad \text{for}\quad n\rightarrow\infty
	\end{equation}
	if the closure of $A+B$ generates a contractive $C_0$-semigroup. The books \cite[Appx.~D]{Zagrebnov.2019,Reed.1980} provide a comprehensive overview of the extensive literature on the topic,. Here, only closely related results are covered.
	
    The seminal work by Sophus Lie in 1875 proves Equation \ref{eq:trotter} as well as a symmetrized product formula (see Eq.~\ref{eq:symmetrized-Lie}) for real matrices (Lie's product formula). Following the symmetrization, which improves the convergence rate to $n^{-2}$, \citeauthor{Suzuki.1976} investigated different families of product formulas achieving higher powers in the convergence rate \cite{Suzuki.1976,Suzuki.1993,Suzuki.1997,Hatano.2005}.\\
    In the direction of the mentioned result by \citeauthor{Trotter.1959} for $C_0$-semigroups, \citeauthor{Kato.1974} proved a convergence result for non-negative self-adjoint operators on a separable Hilbert space with weak assumptions on the common domain as well as a generalization beyond semigroups \cite{Kato.1974,Kato.1978}.\\
    In the direction of Trotter-Suzuki product formulas, the work \cite{Bachmann.2022} proves explicit convergence rate for Hamiltonian evolutions on quantum lattices systems in the infinite volume limit. Recently, the quantitative result \cite{Burgarth.2023} addresses Trotter-Suzuki formulas in infinite closed quantum systems, leveraging knowledge of the eigenstates of the limiting dynamics.\\
    Combining the search for a good convergence rate and the generalization to $C_0$-semigroups, many interesting works have investigated the interplay of the domain of the generators with the convergence rate in the uniform operator topology \cite{Rogava.1993,Neidhardt.1998,Ichinose.1998,Tamura.2000,Neidhardt.1999,Zagrebnov.2022sumrate} and recently in the strong operator topology \cite{Becker.2024, Vanluijk.2024}. Additionally, the extension to evolution systems (time-dependent semigroups) can, for example, be found in \cite{Vuillermot.2009,Neidhardt.2019,Ikeda.2023,Stephan.2023}.
	
	The structural relation to the quantum Zeno effect becomes clear at the example of the projective Zeno effect in closed quantum systems introduced by \citeauthor{Beskow.1967}. Here, $T(t,0)=e^{iPHP}P$ and $F_j=Pe^{t\frac{A}{n}}P$ for any Hermitian matrix $H$ and projection $P$ represent the result by \citeauthor{Misra.1977}, who named the effect after the greek philosopher Zeno of Elea,
	\begin{equation}\label{eq:zeno-misra-sudarshan}
		\|(Pe^{\frac{it}{n}H})^n-e^{it\,PHP}P\|\rightarrow0\qquad\text{for}\qquad n\rightarrow\infty\,.
	\end{equation}
	Since then, the Zeno effect was generalized in several different directions (see \cites{Facchi.2008}{Schmidt.2004}{Itano.1990} for an overview). Building upon the research contributions of \cites{Burgarth.2020}{Mobus.2019}{Barankai.2018}, the quantum Zeno effect has been extended to encompass open and infinite-dimensional quantum systems, featuring general quantum operations but uniformly continuous time evolutions. Especially, we allowed for evolution systems in the work \cite{Mobus.2019}. Going beyond uniformly continuous semigroups, which are always generated by an bounded operator, to $C_0$-semigroups, \citeauthor{Exner.2021} have recently demonstrated the convergence of the projective Zeno effect in the strong operator topology for unbounded Hamiltonians under the weak assumption that $H^{1/2}P$ is densely defined and $(H^{1/2}P)^\dagger H^{1/2}P$ self-adjoint \cite{Exner.2021}. In open quantum systems, the recent work by \citeauthor{Becker.2021} \cite{Becker.2021} and our follow up \cite{Moebus.2023} show that the optimal convergence rate can be achieved in the strong operator topology: Let $(\cL,\cD(\cL))$ be the generator of a contractive $C_0$-semigroup on a Banach space $\cX$ and $M$ a contraction satisfying the uniform power convergence, i.e.~there is a projection $P$ such that $\norm{M^n-P}_{\cX\rightarrow\cX}\leq\delta^n$. If $P\cL P$ is a generator of a contractive $C_0$-semigroup and the asymptotic Zeno condition, i.e.~there is a $b\geq0$ such that $\|Pe^{t\cL}(\1-P)\|_{\cX\rightarrow\cX}, \|(\1-P) e^{t\cL}P\|_{\cX\rightarrow\cX}\leq tb$ for all $t\ge 0$ holds true, then there is an explicit constant $c(t,b)>0$ independent of $n$ so that for any $t\ge 0$ and all $x\in\cD((\cL P)^2)$ 
	\begin{equation*}
		\norm{\left(Me^{\frac{t}{n}\cL}\right)^nx-e^{tP\cL P}Px}\leq\frac{c(t,b)}{n}\left(\norm{x}+\norm{\cL P x}+\norm{(\cL P)^2x}\right)+\delta^n\norm{x}\,.
	\end{equation*}
	Both works (\cite{Becker.2021,Moebus.2023}) are proven in the setup of Banach spaces and can be applied to open or closed quantum systems choosing the underlying space to be an separable Hilbert space $\cH$ or the space of trace-class operators $\mathcal{T}_1(\cH)$.\\
	As already mentioned, the projective quantum Zeno effect frequently measures an evolving quantum state and thereby forces the evolution to remain within the projective space. This concept is, for example, used to construct the $Z(\theta)$-gate in bosonic error correction \cite{Mirrahimi.2014, Touzard.2018, Guillaud.2019, Guillaud.2023}. Moreover, certain measurements can even decouple the system from its environment \cites{Facchi.2004}{Burgarth.2020} or protect the initial state from decoherence \cites{Dominy.2013}{Erez.2004}.
	
	In this work, we generalize the Trotter product formula and the quantum Zeno effect in the following directions. First, we prove quantitative bounds for both by assuming that there exists a stronger Banach space $\cY$, on which the reduction of the limiting $C_0$-semigroup is well-defined, stable, and relatively bounds all involved generators. This assumption is not only motivated by hyperbolic $C_0$-evolution systems (time-dependent $C_0$-semigroups) but also by so-called Sobolev-preserving semigroups \cite{Gondolf.2023}, which is the main example of the present paper. Nevertheless, we also extend the product formulas to evolution systems and present a large class of bosonic semigroups, for which the results can be easily applied. Moreover, we manage to prove higher-order approximations for the Trotter product formula.
	
	For that, the paper is divided into four parts: First, we briefly recap the necessary notation and statements in Section \ref{sec:preliminaries}. Then, we give an overview over the main results and some examples in Section \ref{sec:main-results}. Next, we discuss the results and proofs for Trotter-like product formulas in Section \ref{sec:trotter}. Here, we begin with a bound for the Trotter product formula in the strong operator topology and extend this to $C_0$-evolution systems and Suzuki's higher-order approximations. We then apply the Trotter-like results to Sobolev-preserving semigroups. Finally, we generalize the quantum Zeno effect in Section \ref{sec:zeno}, starting with a bound for the projective Zeno effect, followed by an extension to $C_0$-evolution systems and more general measurements. Once again, the results are applied to Sobolev-preserving semigroups, particularly, the $Z(\theta)$-gate in bosonic error correction.
	
\newpage

\section{\bfseries Preliminaries}\label{sec:preliminaries}
	We start  by revisiting certain tools from Banach space theory used in the present paper. More details can be found in \cites[Ch.~III]{Kato.1995}[Ch.~IV]{Conway.2007}[Ch.~2-3]{Simon.2015}[Ch.~2-3]{Hille.1996}. Let $(\cX,\|\cdot\|_{\cX})$ and $(\cY,\|\cdot\|_{\cY})$ denote Banach spaces and $\cB(\cX,\cY)$ the space of bounded linear maps from $\cX$ to $\cY$, which defines a Banach space again if it is equipped with the operator norm:
	\begin{equation}\label{eq:operator-norm}
		\|A\|_{\cX\rightarrow\cY}\coloneqq\sup_{\|x\|_{\cX}=1}\|A(x)\|_{\cY}\qquad\text{for}\,A\in\cB(\cX,\cY)\,.
	\end{equation}
	We denote the identity map in $\cB(\cX,\cY)$ by $\1$ and the space of bounded endomorphisms by $\cB(\cX)$. More generally, an unbounded operator $A:\cD(A)\subset\cX\rightarrow\cX$ is a linear map defined on a domain $\cD(A)$ and is called densely defined if its domain is a dense subspace of $\cX$. The addition and concatenation of two unbounded operators $(A,\cD(A))$ and $(B,\cD(B))$ is defined on $\cD(A+B)=\cD(A)\cap\cD(B)$ and $\cD(AB)=B^{-1}(\cD(A))$. An operator $(A,\cD(A))$ is closed iff its graph $\{(x,A(x))\,|\,x\in\cD(A)\}$ is a closed set in the product space $\cX\cross\cX$. By convention, we extend all densely defined and bounded operators by the bounded linear extension theorem to bounded operators on $\cX$ \cite[Thm.~2.7-11]{Kreyszig.1989}. Here, $\overline{A}$ is an extension of $A$ if $\cD(A)\subset\cD(\overline{A})$ and $Ax=\overline{A}x$ for all $x\in\cD(A)$. A densely defined operator is called closable if there is a closed extension --- its smallest is called closure. Given operators\footnote{As a short reminder, in this work we use the convention that an operator is a linear map on a Banach space.} $(A,\cD(A)), (B,\cD(B))$ on $\cX$, the operator $(B,\cD(B))$ is relatively $A$-bounded if $\cD(B)\subseteq \cD(A)$ and there are $a,b \geq 0$ for all $x\in\cD(A)$ such that
    \begin{equation}\label{eq:relative-bounded}
        \|B(x)\|_{\cX}\leq a\|A(x)\|_{\cX}+b\|x\|_{\cX}\,.
    \end{equation}
	Similarly, we call $(A,\cD(A))$ relatively $\cY$-bounded w.r.t.~a Banach space $\cY\subset\cD(\cA)$ if $(\cY,\|\cdot\|_{\cY})$ is continuously embedded, i.e.~there is a $b\geq0$ such that $\|x\|_{\cX}\leq b\|x\|_{\cY}$ for all $x\in\cY$, and 
	\begin{equation}\label{eq:relative-Y-bounded-op}
		\|A\|_{\cY\rightarrow\cX}<\infty
	\end{equation}
	Besides the convergence w.r.t.~the operator norm (i.e.~uniform convergence), a sequence of operators $\{A_k\}_{k\in\N}$ defined on a common domain $\cD(A)$ converges strongly to $(A,\cD(A))$ if 
	\begin{equation*}
		\lim_{k\rightarrow\infty}\|A_kx-Ax\|_{\cX}=0
	\end{equation*}
	for all $x\in\cD(A)$. Next, specific weighted Banach spaces are introduced, which are important for the examples of this work. Let $(\cW,\cD(\cW))$ be an injective operator, then $\cD(\cW)$ equipped with
	\begin{equation}\label{eq:weighted-norm}
		\|X\|_{\cW}\coloneqq\|\cW(X)\|_{\cX}
	\end{equation}
	defines a normed space, which is additionally complete if $(\cW,\cD(\cW))$ is surjective. This can be seen by the closed Graph theorem, which shows that $\cW$ is a closed operator \cite[Lem.~2.2]{Gondolf.2023}.
	
	In our examples, $\cW$ plays the role of a reference operator which relatively bounds the considered generators and  defines a stable or admissible subspace required for our theorems.

	\subsection{Strongly continuous semigroups}\label{subsec:c0-semigroups}
		A quantum mechanical evolution is often described by a differential equation called the master equation, for which strongly continuous ($C_0$-)semigroups and $(C_0)$-evolution systems constitute an important toolbox for its solution theory. More details can for example be found in \cite[Chap.~II]{Engel.2000}\cite[Chap.~9]{Kato.1995} or \cite[Chap.~X-XIII]{Hille.1996}. A $C_0$-semigroup is defined by a family of operators $(T_t)_{t\geq0}\subset\cB(\cX)$, which satisfies
		\begin{equation}\label{eq:semigroup}
			T_tT_s=T_{t+s},\quad T_0=\1,\quad\text{and}\quad \lim\limits_{t\downarrow 0}T_tx=x
		\end{equation}
		for all $t,s\geq0$ and $x\in\cX$. Equivalently, every semigroup is uniquely determined by its generator \cite[Thm.~II.1.4]{Engel.2000} --- a possible unbounded, densely defined, and closed operator given by
		\begin{equation}
    		\cL(x) = \lim\limits_{t \to 0^+} \frac{1}{t}(T_t(x) - x) \, , 
		\end{equation}
		for any 
		\begin{equation}
			x\in\cD(\cL) = \{x \in \cX : t \mapsto T_t(x) \text{ differentiable on } \R_{\geq0}\}\,.
		\end{equation}
		This justifies the notation $T_t=e^{t\cL}$. Interestingly, one can show that every $C_0$-semigroup satisfies
		\begin{equation}\label{eq:stable-semigroup}
			\|T_t\|_{\cX\rightarrow\cX}\leq ce^{\omega t}
		\end{equation}
		for a $c\geq0$ and $\omega\geq0$ \cite[Prop.~I.5.5]{Engel.2000}, which is referred to as stability of the semigroup. In the special case $c\leq1$ and $\omega=0$, the semigroup is called contractive. Since the semigroup leaves the domain $\cD(\cL)$ invariant and commutes with its generator, the following master equation admits the unique solution $T_t(x(0)) = x(t)$ \cite[Prop.~II.6.2]{Engel.2000}
		\begin{equation}\label{eq:time-indep-master-equation}
    		\frac{d}{dt} x(t) = \cL(x(t)) \quad x(0) \in \cD(\cL) \quad\text{and}\quad t \ge 0 \, . 
		\end{equation}
		The above equation can be integrated in the following sense. For $[a,b]\subset[0,T]$, we denote the Bochner integral (see \cites[Sec.~3.5-8]{Hille.1996}) of a vector-valued function $[0,T]\ni t\mapsto x(t)\in\cX$ by
		\begin{equation*}
			\int_{a}^b x(t)dt\,.
		\end{equation*}
		If the integral is well-defined, the usual properties like linearity, triangle inequality, additivity of disjoint sets, dominated convergence, as well as invariance under closed operators $A$ if $[a,b]\ni t\mapsto Ax(t)\in\cX$ is integrable, i.e. 
		\begin{equation}
			A\int_{a}^bx(t)d\mu=\int_{a}^bAx(t)dt\,,
		\end{equation}
		can be shown \cite[Sec.~3.7]{Hille.1996}. Moreover, the Fundamental theorem of calculus can be proven:
		\begin{equation}\label{eq:fund-thm-calculus}
			x(b)-x(a)=\int_{a}^{b}\frac{d}{dt}x(t)dt\,.
		\end{equation}
		One way to verify that a vector-valued function on a compact interval is integrable is to prove continuity of the function \cite[Thm.~3.7.4]{Hille.1996}. In the present work (see for instance Thm.~\ref{thm:trotter-tele}, \ref{thm:general-zeno}), the vector-valued map is usually given by a semigroup, an evolution system (see next paragraph), or a combination of both $t\mapsto G(t)x$ for $x\in\cD(G(t))$. If the $G$ is not just a semigroup which is by definition strongly continuous (see \cite[Lem.~2.1.3]{Engel.2000}), we present the following slight extension of \cite[Lem.~B.16]{Engel.2000} as an auxiliary tool for the later proofs:
		\begin{lem}\label{lem:product-continuity}
			Let $\cY\subset\cX$ be a continuously embedded subspace of $\cX$, $[0,1]\ni t\mapsto G_1(t)\in\cB(\cX)$ and $[0,1]\ni t\mapsto (G_2(t),\cY)$ be two strongly continuous vector-valued maps. Then
			\begin{equation*}
				[0,1]\ni t\mapsto (G_1(t)G_2(t),\cY)
			\end{equation*}
			is again a strongly continuous vector-valued map on $\cY$. 
		\end{lem}
		\begin{proof}
			The proof follows the proof of \cite[Lem.~B.15]{Engel.2000}. Let $x\in\cY$, $t\in[0,1]$, and $(t_n)_{n\in\N}$ a sequence converging to $t$. Then, there is an uniform bound on $\|G_1(t_n)\|_{\cX\rightarrow\cX}$ by the uniform boundedness theorem. Therefore, 
			\begin{equation*}
				\begin{aligned}
					\|G_1(t_n)G_2(t_n)x-G_1(t)G_2(t)x\|_{\cX}&=\|G_1(t_n)\|_{\cX\rightarrow\cX}\,\|G_2(t_n)x-G_2(t)x\|_{\cX}\\
					&\qquad+\|(G_1(t_n)-G_1(t))G_2(t)x\|_\cX
				\end{aligned}
			\end{equation*}
			finishes the proof.
		\end{proof}
		
		As  mentioned, we consider so-called evolution systems. These extend the master equation in Equation \ref{eq:time-indep-master-equation} by assuming that the generator is time-dependent. In this case, a two-parameter family of bounded operators $( T_{t,s})_{(t,s)\in I_T}$ with $I_{T}\coloneqq\{(t,s)\,|\,T\geq t\geq s\geq0\}$ is called a $C_0$-evolution system (or time-dependent $C_0$-semigroup) if
		\begin{equation}\label{eq:evolution-sys}
			T_{t,t}=\1\,,\quad T_{t,r}T_{r,s}=T_{t,s}\,,\quad\text{and}\quad (t,s)\mapsto T_{t,s}x\text{ is strongly continuous}
		\end{equation}
		for all $0\leq s\leq r\leq t\leq T$. Clearly, every $C_0$-semigroup $T_t$ can be identified by an evolution system via $T_{t,s}=T_{t-s}$. However, a crucial difference is that evolution systems are not necessarily differentiable for any $x\neq 0$ \cites[p.~478]{Engel.2000}. However, in the present paper all considered evolution systems are defined to be differentiable in the sense that there is a family of generators $(\cL_s,\cD(\cL_s))_{0\leq s\leq T}$ such that for all $x\in\cD(\cL_s)$ and all $(t,s)\in I_T$,
		\begin{equation}\label{eq:evo-derivatives}
			\begin{aligned}
				\frac{\partial}{\partial s}T_{t,s}x=-T_{t,s}\cL_sx\,.
			\end{aligned}
		\end{equation}
		We say that $(\cL_s,\cD(\cL_s))_{0\leq s\leq T}$ generates $T(t,s)$ if it is uniquely determined by the above relation. For a continuously embedded Banach space $(\cY,\|\cdot\|_{\cY})$ in $\cX$, we call a continuously differentiable function $t\mapsto x(t)\in\cY$ for $(t,s)\in I_T$ a $\cY$-solution if it satisfies the differential equation
		\begin{equation}\label{eq:time-dep-master-equation}
    		\frac{\partial}{\partial t}x(t)=\cL_t(x(t))\quad\text{and}\quad x(s)=x_s\in\cY\quad\text{ for }0\leq s\leq t\leq T
		\end{equation}
        for a given family $(\cL_t,\cY\subset\cD(\cL_t))$ with $(t,s)\in I_T$. Conversely, we introduce sufficient assumptions on $(\cL_s,\cD(\cL_s))$ such that the differential Equation \ref{eq:time-dep-master-equation} admits a unique solution operator \cites[Chap.~5]{Pazy.1983}[pp.~127~ff.]{Giuseppe.1976}. For that, we introduce the notion of admissible subspaces:
		\begin{defi}[Stable and admissible subspaces]\label{defi:admissible-spaces}
			Let $(\cY\subset\cX,\|\cdot\|_{\cY})$ be a continuously embedded subspace of $\cX$ and $(\cL,\cD(\cL))$ the generator of a $C_0$-semigroup on $\cX$.
			Then, the subspace $(\cY,\|\cdot\|_{\cY})$ is $\cL$-stable, if it is invariant under the semigroup, i.e.~$e^{t\cL}\cY\subset\cY$ for all $t\geq0$, and the reduced vector-valued function $e^{t\cL}|_\cY:\cY\rightarrow\cY$ is stable, i.e.~there are $\tilde{c},\tilde{\omega}\geq0$ such that 
			\begin{equation*}
				\|e^{t\cL}\|_{\cY\rightarrow\cY}\leq \tilde{c}e^{t\tilde{\omega}}\,.
			\end{equation*}
		 	Moreover, $(\cY,\|\cdot\|_{\cY})$ is $\cL$-admissible if it is $\cL$-stable and the reduced function is strongly continuous. The same definitions translates to $C_0$-evolution systems with time-independent constants. 
		\end{defi}
		\noindent Note that the additional continuity assumption in the definition of $\cL$-admissible subspaces induces, that the reduced vector-valued function, i.e.~$e^{t\cL}|_\cY:\cY\rightarrow\cY$, defines a semigroup again. Here, the stability of the reduced vector-valued map is not required and even implied by \cite[Prop.~I.5.5]{Engel.2000}. However, in the following, we consider families of generators over $[0,T]$ for which we assume to have uniform stability constants $\tilde{c},\tilde{\omega}$ for all $s\in[0,T]$. The describing properties, which refer to evolution systems of hyperbolic type, are:
		\renewcommand{\labelenumi}{(\arabic{enumi})\hspace{1ex}}
		\begin{enumerate}
			\item For every $0\leq s\leq T$, $(\cL_s, \cD(\cL_s))$ generates a $C_0$-semigroup and admits uniform stability constants $c,\omega\geq0$, i.e. $\|e^{t\cL_s}\|_{\cX\rightarrow\cX}\leq ce^{t\omega}$ for all $t\geq0$ and $s\in[0,T]$\,;
			\item There exists a $\cL_s$-admissible subspace $(\cY\subset \cX,\|\cdot\|_{\cY})$ and uniform stability coefficients $\tilde{c},\tilde{\omega}$ for all $0\leq s\leq T$;
			\item The map $s\mapsto\cL_s$ is $\cY$-bounded (see Eq.~\ref{eq:relative-Y-bounded-op}) and continuous in $\cB((\cY,\|.\|_\cY),(\cX,\|.\|_\cX))$.
		\end{enumerate}
		To achieve a $\cY$-solution of the differential equation (\ref{eq:time-dep-master-equation}), one can additionally assume 
		\begin{enumerate}
			\setcounter{enumi}{3}
			\item  $T_{t,s} \cY \subseteq \cY$ for $(t,s)\in I_T$; 
			\item $(s, t) \mapsto T_{s,t}$ is strongly continuous on $(\cY, \norm{\cdot}_{\cY})$\,.
		\end{enumerate}
		Under these assumptions, one can show: 
		\begin{thm}[Hyperbolic evolutions systems \texorpdfstring{\cite[Thm.~5.3.1/4.3]{Pazy.1983}}{???}]\label{thm:time-dependent-semigroups}
			Let $(\cL_s,\cD(\cL_s))_{s\ge 0}$ be a family of operators satisfying assumptions $(1-3)$. Then, there is a unique hyperbolic $C_0$-evolution system on $(\cX,\|\cdot\|_{\cX})$ satisfying
			\begin{equation*}
				\|T_{t,s}\|_{\cX\rightarrow\cX}\leq c\,e^{(t-s)\omega}\,,\quad \lim\limits_{t\downarrow s}\frac{T_{t,s}x-x}{t-s}=\cL_sx\,,\quad\text{and}\quad \frac{\partial}{\partial s}T_{t,s}x=-T_{t,s}\cL_sx
			\end{equation*}
			for all $(t,s)\in I_T$ and $x\in\cY$. If further $(4)$ and $(5)$ hold, then for every $x_s\in \cY$, $T_{t,s}x_s$ is a unique solution for the initial value problem in (\ref{eq:time-dep-master-equation}).
		\end{thm}
		\begin{proof}
			To apply the result \cite[Thm.~5.3.1/4.3]{Pazy.1983}, one can directly translate the stability assumptions given in $(1)-(5)$ via the Feller-Miyadera-Phillips generation theorem (see \cite[Thm.~III.3.8]{Engel.2000}) and \cite[Prop.~I.5.5]{Engel.2000} to the assumptions $H_1,...,H_3$ and $E_4,\,E_5$ given in \cite[p. 135,141]{Pazy.1983} as well as the definition of stability in \cite[Def.~5.2.1]{Pazy.1983}.
		\end{proof}
	
	\subsection{Bosonic Sobolev preserving semigroups}\label{subsec:bosonic-semigroups}
		One of the main examples in our work is the class of bosonic Sobolev preserving semigroups, which are discussed in detail in \cite{Gondolf.2023}. For those, we shortly introduce bosonic (or continuous variable) systems in the second quantization formalism \cite[Chap.~5.2]{Bratteli.1981}. Let $\cH=L^2(\mathbb{R}^m)$ be the Fock space of $m$-modes endowed with an orthonormal (Fock-)basis $\{\ket{\n}\}_{\n\in\N^m}$, where $n_j$ labels the number of particles in the $j^{\text{th}}$-mode. Then, the {annihilation} and {creation} operators are defined by
		\begin{equation*}
			a_j\ket{\n}=\sqrt{n_j}\ket{n_1,...,n_{j}-1,...,n_m},\quad\ad_j\ket{\n}=\sqrt{n_j+1}\ket{n_1,...,n_j+1,...,n_m}
		\end{equation*}
		and satisfy
		\begin{equation*}
			[a_i,\ad_j]\ket{\n}=\delta_{i,j}\ket{\n}\,,\quad[a_i,a_j]\ket{\n}=0
		\end{equation*}
		for all $\n\in\N^m$ and $i,j\in\{1,...,m\}$. Moreover, the number operator is defined by
		\begin{equation}\label{eq:number-operator}
			N_j\ket{\n} = \ad_j a_j\ket{\n} =n_j\ket{\n}
		\end{equation}
		and counts the number of particles in the $j^{\text{th}}$-mode. Another important class are the coherent states, which are eigenvectors of $a_j$, i.e.~$a_j\ket{\alpha}=\alpha_j\ket{\alpha}$ for all $j\in\{1,...,m\}$. Here, the $j^{\text{th}}$-mode is defined by $\ket{\alpha_j}=\exp(-\frac{|\alpha_j|^2}{2})\sum_{j=0}^\infty\frac{\alpha_j^n}{\sqrt{n!}}\ket{n_j}$ for any $\alpha_j\in\C$. In most of our proofs, we use the finitely supported subspace $\cH_f = \{\ket{\psi}= \sum_{\n=0}^M \lambda_{\n}\ket{\n}\,|\,M\in\N\,,\,\n\in\N^m\,,\,\lambda_{\n}\in\C\}$ as a domain for $a_j$, $\ad_j$, and polynomials in both and extent the bound afterwards by continuity in an appropriate topology. 
		
		For $A\in\cB(\cH)$, the adjoint is uniquely defined by $\braket{A\phi\,,\varphi}=\braket{\phi\,,A^\dagger\varphi}$ for all $\phi,\varphi\in\cH$ \cite[Sec.~2.1]{Simon.2015} and the space of self-adjoint operators denoted by $\cB_{\operatorname{sa}}(\cH)\coloneqq\{A\in\cB(\cH)\,|\,A=A^\dagger\}$. Then, we define the trace by $\tr[A]=\sum_{\n\in\N^m}\bra{\n}A\ket{\n}$ for a positive $A\in\cB_{\operatorname{sa}}(\cH)$, i.e.~there is a $B\in\cB(\cH)$ such that $A=B^\dagger B$, and the trace norm by $\|A\|_1=\tr[\sqrt{A^\dagger A}]$. Finally, the space of self-adjoint trace-class operators is denoted by $\cT_{1,\operatorname{sa}}\coloneqq\{A\in\cB_{\operatorname{sa}}(\cH)\,|\,\|A\|_1<\infty\}$. Using the convention $\n\leq M$ iff $n_j\leq M$ for all $j\in\{1,...,m\}$, we define 
		\begin{equation}\label{eq:finite-approx-states}
			\cT_f\coloneqq \{A = \sum_{\n,\m=0}^M a_{\n\m} \ketbra{\n}{\m}\,|\,M\in\N\,,\,\,\n,\m\in\N^m\,,\,\,a_{\n,\m}=\bar{a}_{\m,\n}\in\C\}
		\end{equation}
		as before, which is dense in $\cT_{1,\operatorname{sa}}$ \cite[Sec.~3.1]{Simon.2015}. Then, we define the bosonic Sobolev spaces by
		\begin{equation}\label{eq:bosonic-symmetric-weight}
			\cD(\cW^\k)=\cW^{-\k}(\cT_{1,\,\operatorname{sa}})\quad\text{via}\quad\cW^{-\k}(x)\coloneqq\prod_{j=1}^m(\1+N_j)^{-k_j/4}\,\,x\,\prod_{j=1}^m(\1+N_j)^{-k_j/4}\,.
		\end{equation}
		Since $\cW^{-\k}$ is an injective bounded operator, its inverse $\cW^{\k}$ is a closed operator on $\cD(\cW^\k)$ by the closed graph theorem and $(\cD(\cW^\k), \|\cdot\|_{\cW^\k})$ a Banach space. For the sake of notation, we define
		\begin{equation}\label{eq:sobolev-space}
			W^{\k,1}\coloneqq\cD(\cW^\k)\quad\text{and}\quad\|\cdot\|_{W^{\k,1}}\coloneqq\|\cdot\|_{\cW^\k}
		\end{equation}
		with the special case $(W^{0,1},\|\cdot\|_{W^{0,1}})=(\cT_{1,\,\operatorname{sa}},\|\cdot\|_1)$. Interestingly, one can show that if $0\leq k_j < k_j'$ for all $j\in\{1,...,m\}$, $W^{\k',1}$ is compactly embedded in $W^{\k,1}$ \cite{Gondolf.2023}.

		Next, we discuss a sufficient condition so that an operator $(\cL,\cT_f)$ is a generator. Note that the considered operators are often defined on the domain $\cT_f$ (see Eq.~\ref{eq:finite-approx-states}) and then extended to their closure if possible. First, we assume that the operator is in $\operatorname{GKSL}$ form, i.e.
        \begin{equation}\label{eq:lindblad}
			\cL:\cT_f \to \cT_f \quad x \mapsto\cL(x) = - i [H, x] + \sum\limits_{j = 1}^K L_j x L_j^\dagger  - \frac{1}{2}\{L_j^\dagger L_j, x\}
		\end{equation}
		with $K\in\mathbb{N}$, $L_j\coloneqq p_j(a_1,\ad_1,...)$ defined by polynomials of $a_j$ and $\ad_j$ for all $j\in\{1,...,m\}$, and symmetric $H\coloneqq p_H(a_1,\ad_1,...)$, i.e.~$\braket{A\psi, \phi}=\braket{\psi, A\phi}$ for all $\psi,\phi\in\cH_f$. Here, we assume that the monomials in each mode are ordered like $(\ad_j)^k a_j^l$. Second, we assume that there are non-negative sequences $(\k_r)_{r \in \N}\subset\N^m$ diverging coordinate-vice to infinity and $(\omega_{\k_r})_{r\in\N}$ so that 
		\begin{equation}\label{eq:assumsobolevstability}
			\tr[\cW^{\k_r}\cL(x)] \le \omega_{\k_r} \tr[\cW^{\k_r}(x)]
		\end{equation}
		for all positive semi-definite $x \in \cT_f$. Then, the following theorem holds: 
		\begin{thm}[\texorpdfstring{\cite[Thm.~4.4]{Gondolf.2023}}{???}]\label{thm:multimode-generation-theorem}
			Let $(\cL, \cT_f)$ be an operator of the form (\ref{eq:lindblad}) satisfying Equation \ref{eq:assumsobolevstability}, then the closure $\overline{\cL}$ generates a positivity preserving $C_0$-semigroup $(T_t)_{t\ge 0}$ on $W^{\k_r, 1}$ for $r \in \N$ and
			\begin{equation}\label{eq:sobolev-preserving-ineq}
				\norm{T_t}_{W^{\k_r, 1} \to W^{\k_r, 1}} \le e^{\omega_{\k_r} t}
			\end{equation}
			for all $t\geq 0$. In the special case $\k = 0$, the semigroup is contractive and trace-preserving.
		\end{thm}
		The same holds true if the coefficients of the operator are time-dependent but continuous and if the right-hand side in Equation \ref{eq:assumsobolevstability} does not depend on $t$ \cite[Thm.~4.5]{Moebus.2023}.
		
\section{\bfseries Main results}\label{sec:main-results}
	This section presents the main results of the paper and discusses its assumptions. The paper is split into two parts: The first is about the Trotter product formula and the second about the quantum Zeno effect. We start with the former and a generalization \citeauthor{Trotter.1959} published in \citeyear{Trotter.1959} of Lie's product formula for matrices to contractive $C_0$-semigroups on a Banach space $\cX$ generated by $(\cL,\cD(\cL))$ and $(\cK,\cD(\cK))$. In this section, we assume that $(\cL,\cD(\cL))$ and $(\cK,\cD(\cK))$ generates a contractive $C_0$-semigroup, without explicitly mentioning. Moreover, assume that the closure of $\cL+\cK$ generates a contractive $C_0$-semigroup. Then, \cite{Trotter.1959} shows that for all $x\in\cX$
	\begin{equation}\label{eq:lie-trotter}
		\left(e^{\frac{\cL}{n}}e^{\frac{\cK}{n}}\right)^nx\rightarrow e^{\cL+\cK}x\qquad\text{for}\qquad n\rightarrow\infty\,.
	\end{equation}
	For our results, we need two main assumptions. First, we assume that there is a (reference) Banach space $(\cY\subset\cX,\|\cdot\|_{\cY})$ on which the operators $\cL$ and $\cK$ as well as all products are not only well-defined but also relatively $\cY$-bounded. Repeat Equation \ref{eq:relative-Y-bounded-op}: an operator $(A,\cD(A))$ is relatively $\cY$-bounded if $(\cY,\|\cdot\|_{\cY})$ is continuously embedded in $(\cX,\|\cdot\|_{\cX})$ and 
	\begin{equation}\label{assum1:relative-boundedness}
		\|A\|_{\cY\rightarrow\cX}<\infty\,.
	\end{equation}
	Second, we assume that the closure of $(\cL+\cK)$ generates a contractive $C_0$-semigroup on $\cX$ and a stable vector-valued map on $\cY$ --- $\cY$ is $(\cL+\cK)$-stable --- which means $e^{t(\cL+\cK)}(\cY)\subset\cY$ and	there are a $\tilde{c},\omega\geq 0$ so that
	\begin{equation}\label{assum2:stability}
		\|e^{t(\cL+\cK)}\|_{\cY\rightarrow\cY}\leq \tilde{c}e^{t\omega}\,.
	\end{equation}
	Before diving into the first result, we briefly discuss the motivation behind the above assumptions. These assumptions embody the essential steps used in the current proof strategy. For instance, they cover the following examples. First, we consider an operator-theoretical example assuming certain relative boundedness assumption between the involved generators. 
	\begin{ex}\label{ex:back-to-operator-rel-boundedness}
		Assume that the closure of $(\cL+\cK,\cD(\cL+\cK))$ defines a contractive $C_0$-semigroup again and that $\cK^2$, $\cL\cK$, $\cK \cL$, and $\cL^2$ are relatively $(\overline{\cL+\cK})^2$-bounded. Then, we equip $\cY=\cD((\overline{\cL+\cK})^2)$ with the norm
		\begin{equation*}
			\|x\|_{\cY}\coloneqq \|x\|_{\cX}+\|(\overline{\cL+\cK})x\|_{\cX}+\|(\overline{\cL+\cK})^2x\|_{\cX}
		\end{equation*}
		for $x\in\cY$, which defines a Banach space by repeating the reasoning why $\cD(\cL)$ with the graph norm defines a Banach space twice \cite[Sec.~5.2]{Kato.1995}. Specially, it defines an $(\cL+\cK)$-admissible subspace. Therefore, it satisfies the assumptions Assumption \ref{assum1:relative-boundedness} and \ref{assum2:stability} with $\omega=0$ so that the bound does not depend exponentially on $t$. Note that the above construction is equivalent to the so-called Sobolev Tower of order $2$ \cite[Sec.~5]{Engel.2000} if $\overline{\cL+\cK}$ is invertible.
	\end{ex}
	Second, we consider a bosonic example of decoupled quadratic generators: 
	\begin{ex}\label{ex:bosonic-quadratic-generators}
		Let $(H_1,\cH_f)$ and $(H_2,\cH_f)$ be two decoupled quadratic Hamiltonian in annihilation and creation operators defined on an $m$-mode bosonic quantum system admitting the structure
        \begin{equation}\label{eq-main:quadratic-hamiltonian}
			H=\sum_{i\in\{1,..,m\}}\sum_{k_i+\ell_i\leq2} \lambda_{k_i\ell_i}(\ad_i)^{k_i}a_i^{\ell_i}+p(N_1,...,N_m)
		\end{equation}
        Then, it is shown in \cite[Lem.~30]{Moebus.2023Learning} that there is a constant $\tilde{c}$ such that for all $x\in\cD(W^{\k,1})$
		\begin{equation*}
			\|-i[H_j,x]\|_1\leq\tilde{c}\|\cW^4(x)\|_1\qquad\text{and}\qquad \|-[H_i[H_j,x]]\|_1\leq\tilde{c}\|\cW^{8}(x)\|_1\,,
		\end{equation*}
		where $W^{k,1}=\prod_{j=1}^m(\1+N_j)^{k/4}\,\,x\,\prod_{j=1}^m(\1+N_j)^{k/4}$ and $i,j\in\{1,2\}$. Following the proof of Lemma 5.6 and applying Theorem 4.4 (restated in Thm.~\ref{thm:multimode-generation-theorem}) in \cite{Gondolf.2023}, the closures of the Hamiltonians define quantum Markov semigroups with the property
		\begin{equation*}
			\|e^{-it[H_j,\cdot]}x\|_{W^{k, 1}}\leq e^{\omega t}\|x\|_{W^{k, 1}} 
		\end{equation*}
		for a constant $\tilde{\omega}\geq0$, $k\in\{4,8\}$, $j\in\{1,2\}$. Here, we used the notation for Sobolev spaces introduced in Equation \ref{eq:sobolev-space}. Therefore, the above assumptions \ref{assum1:relative-boundedness} and \ref{assum2:stability} are satisfied for $(\cY,\|\cdot\|_{\cY})=(W^{4,1},\|\cdot\|_{W^{4,1}})$. More details and open bosonic quantum systems are given in Section \ref{subsec:bosonic-trotter}.
	\end{ex}
	Keeping these two examples in mind, we prove an explicit convergence rate for the Trotter product formula in the strong-operator topology.
	\begin{thm*}[See Theorem \ref{thm:trotter-tele}]
		Let $(\cY\subset\cX,\|\cdot\|_{\cY})$ be a Banach space such that $\cY\subset\cD(\cL+\cK)$ and the closure of $(\cL+\cK,\cY)$ generates a contractive $C_0$-semigroup on $\cX$. If additionally $\cK$, $\cL$, $\cK^2$, $\cL\cK$, $\cK\cL$, $\cL^2$ are relatively $\cY$-bounded (see Eq.~(\ref{eq:relative-Y-bounded-op})) and  $(\cY,\|\cdot\|_{\cY})$ is $(\cL+\cK)$-stable, then there are $c,\omega\geq0$ such that for all $x\in\cY$, $t\geq0$, and $n\in\N$
		\begin{equation*}
			\begin{aligned}
				\norm{\left(e^{\frac{t}{n}\cL}e^{\frac{t}{n}\cK}\right)^nx-e^{t(\cL+\cK)}x}_{\cX}&\leq c\frac{t^2}{n}e^{t\omega}\|x\|_{\cY}\,.
			\end{aligned}
		\end{equation*}
	\end{thm*}
	With respect to Example \ref{ex:bosonic-quadratic-generators}, we verified that two decoupled bosonic quadratic Hamiltonian satisfy Trotter's product formula with the convergence rate $n^{-1}$ in the strong operator topology.  
	
	\begin{rmk*}[Variation of time-steps]
		Note that the proof of the above theorem (also the following theorems) holds even for time-steps that are not equally distributed, thanks to the generality of Proposition \ref{prop:telescopic-sum} and the semigroup property of the limiting semigroup. In this context, it suffices for the largest time-step to converge to zero asymptotically faster than order  $\frac{1}{\sqrt{n}}$ in the strong topology. This directly affects the overall convergence rate of the product formula. 
	\end{rmk*}
	
	In the next step, we improve the convergence rate similar to the matrix, in which the symmetrized Lie product formula demonstrates
	\begin{equation}\label{eq:symmetrized-Lie}
		\|\left(e^{\frac{\cL}{2n}} e^{\frac{\cK}{n}} e^{\frac{\cL}{2n}}\right)^n-e^{\cL+\cK}\|_{\cX\rightarrow\cX}=\cO\Bigl(\frac{1}{n^2}\Bigr)
	\end{equation}
	for matrices $\cL,\cK\in\cB(\cX)$ on a finite dimensional Banach space $\cX$. Here, the analysis to achieve a convergence rate of order $n^{m}$ for an $m\in\N$ can be reduced to the bound 
	\begin{equation}\label{eq:suzuki-assumption}
		\|F(t)-e^{t(\cL+\cK)}\|_{\cY\rightarrow\cX}\leq b t^{m+1}\,,
	\end{equation}
	where $F(t)\coloneqq e^{tp_1\cL}e^{tp_2\cK}e^{tp_3\cL}\cdots e^{tp_{m-1}\cL}e^{tp_m\cK}$ with coefficients $p_1,...,p_m\in\R$ (see Lem.~\ref{lem:general-suzuki}). In the special case of the symmetrized Trotter product formula, explicit upper bounds on Equation \ref{eq:suzuki-assumption} are proven to apply Lemma \ref{lem:general-suzuki}: 
	\begin{thm*}[See Theorem \ref{thm:symmetrized-Lie}]
		Let $(\cY\subset\cX,\|\cdot\|_{\cY})$ be a Banach space such that the closure of $(\cY,\cL+\cK)$ generates a contractive $C_0$-semigroup in $\cX$. If additionally all products of $\cK$ and $\cL$ up to the third power\footnote{$\cL$, $\cK$, $\cL^2$, $\cK^2$, $\cL\cK$, $\cK\cL$, $\cK^3$, $\cK^2\cL$, $\cK\cL^2$, $\cL^3$, $\cL^2\cK$, $\cK\cL\cK$, $\cL \cK^2$, $\cL\cK\cL$} are relatively $\cY$-bounded and $(\cY,\|\cdot\|_{\cY})$ is $(\cL+\cK)$-stable, then for all $x\in\cY$ and $t\geq0$ there are $c,\omega\geq0$ so that
		\begin{equation*}
			\begin{aligned}
				\norm{\left(e^{\frac{t}{2n}\cK}e^{\frac{t}{n}\cL}e^{\frac{t}{2n}\cK}\right)^nx-e^{t(\cL+\cK)}x}_{\cX}&\leq c\frac{t^3}{n^2}e^{t\omega}\|x\|_{\cY}\,.
			\end{aligned}
		\end{equation*}
	\end{thm*}
	One observation of the proof is that higher-order convergence rates need higher order derivatives and thereby a higher regularity of the input state. Utilizing the fact that in \cite[Lem.~30]{Moebus.2023Learning} any polynomial in annihilation and creation operator can be bounded by a $\cW^{k}$ for $k\in\N$, the Trotter-Suzuki  formula \cite{Suzuki.1993} can be applied assuming that the input state is in a regular enough Sobolev space.  However, one has to keep in mind that the evolution needs to be time-reversible for $m\geq3$ (see Eq.~\ref{eq:trotter-suzuki}) as shown in \cite{Suzuki.1991}. 
	
	Finally, we extend the product formula to the time-dependent case. Here, the ideas are roughly the same by exploiting the assumption of a subspace on which the evolution defines a stable function. The limit under consideration is 
	\begin{equation}\label{eq:trotter-evolutions-sys}
		\prod_{j=1}^n U(s_{j},s_{j-1})V(s_{j},s_{j-1})\rightarrow T(t,s)\qquad\text{for}\qquad n\rightarrow\infty\,,
	\end{equation}
	where $0\leq s=s_0<...<s_n=t$ is a partition of $[s,t]$, $U(t,s)$ and $V(t,s)$ are two $C_0$-evolution systems with generator families $(\cL_s,\cD(\cL_s))$ and $(\cK_s,\cD(\cK_s))$, and appropriate assumptions.
	\begin{thm*}[See Theorem \ref{thm:trotter-time-indep}]
		Let $U(t,s)$ and $V(t,s)$ be contractive $C_0$-evolution systems for $(t,s)\in I_T$ and with generator families $(\cL_t,\cD(\cL_t))_{t\in[0,T]}$ and $(\cK_t,\cD(\cK_t))_{t\in[0,T]}$. Moreover, assume that there is a Banach space $(\cY\subset\cX,\|\cdot\|_{\cY})$ so that the closure of $(\cL_t+\cK_t,\cY)_{t\in[0,T]}$ generates a unique contractive $C_0$-evolution system $T(t,s)$ on $\cX$ for $(t,s)\in I_T$. If $\cL_s$, $\cK_s$, $\cL_s\cL_t$, $\cK_s\cL_t$, $\cL_s\cK_t$, $\cK_s\cK_t$ are well-defined on $\cY$, relatively $\cY$-bounded, and $(\cY,\|\cdot\|_{\cY})$ is $(\cL_t+\cK_t)$-stable uniformly for all $(s,t)\in I_T$, in particular $\|T(t,s)\|_{\cY\rightarrow\cY}\leq \tilde{c}e^{(t-s)\omega}$, then there is a $c\geq0$ so that
		\begin{equation*}
			\begin{aligned}
				\norm{\biggl(\prod_{j=1}^n U(s_{j},s_{j-1})V(s_{j},s_{j-1})\biggr)x-T(t,0)x}_{\cX}&\leq c\frac{t^2}{n}e^{t\omega}\|x\|_{\cY}
			\end{aligned}
		\end{equation*}
		for all $x\in\cY$, $t\in[0,T]$, and $s_0=0<\frac{t}{n}<...<\frac{(n-1)t}{n}<t=s_n$.
	\end{thm*}
	\noindent Example \ref{ex:bosonic-quadratic-generators} can be extended to time-dependent Hamiltonians defined by decoupled quadratic Hamiltonians with coefficients depending continuously on time:
 
	\begin{ex}\label{ex:bosonic-quadratic-generators-time-dep}
		Let $s\geq0$ and $(H_1(t),\cH_f)_{t\in[s,T]}$ and $(H_2(t),\cH_f)_{[s,T]}$ be two decoupled quadratic Hamiltonians of the form Equation \ref{eq-main:quadratic-hamiltonian} with time-dependent coefficients on an $m$-mode bosonic quantum system. Then, the same bounds as in Example \ref{ex:bosonic-quadratic-generators} hold true by \cite[Lem.~30]{Moebus.2023Learning} and the closures of $(-i[H_1(t),\cdot],\cT_f)_{t\in[s,T]}$, $(-i[H_2(t),\cdot],\cT_f)_{t\in[s,T]}$ define a Sobolev preserving quantum evolution system \cite[Lem.~5.6, Thm.~4.5]{Gondolf.2023}. Moreover, the assumptions for the above theorem is satisfied (more details and open bosonic quantum systems are given in Section \ref{subsec:bosonic-trotter}).
	\end{ex}

	Similar to the Trotter product formula, the quantum Zeno effect describes an operator product of a contractive projection $P\in\cB(\cX)$, i.e.~$P^2=P$, and a $C_0$-semigroup $e^{t\cL}$:
	\begin{equation}\label{eq:zeno}
		(Pe^{\frac{\cL}{n}})^n\rightarrow e^{P\cL P}P\qquad\text{for}\quad n\rightarrow\infty\,.
	\end{equation}
	Under certain assumptions the convergence as well as the convergence rate can be achieved in the strong operator topology \cite{Moebus.2023,Becker.2021}. We start with the projective Zeno effect without assuming the asymptotic Zeno condition \cite[Eq.~6]{Moebus.2023}. Again, we assume in the following that $(\cL,\cD(\cL))$ is a generator of a contractive $C_0$-semigroup and $P$ a projective contraction.
	\begin{prop}[See Proposition \ref{prop:simple-zeno}]
		Assume there is a Banach space $(\cY,\|\cdot\|_{\cY})$ invariant under $P$ such that the closure of $(P\cL P,\cY)$ defines a contractive $C_0$-semigroup, $\cL P$ and $\cL^2P$ are relatively $\cY$-bounded, and $(\cY,\|\cdot\|_{\cY})$ is a $P\cL P$-stable subspace, in particular $\|e^{tP\cL P}\|_{\cY\rightarrow\cY}\leq \tilde{c}e^{t\omega}$. Then, there is a $c_1\geq0$ such that for all $t\geq0$, $x\in\cY$, and $n\in\N$
		\begin{equation*}
			\begin{aligned}
				\|\left(Pe^{t\frac{\cL}{n}}P\right)^nPx-e^{tP\cL P}Px\|_{\cX}&\leq\frac{c_1t^2}{n}\left(e^{t\omega}\|Px\|_{\cY}+\|(P\cL P)^2x\|_{\cX}\right).
			\end{aligned}
		\end{equation*}
		If we additionally assume that $(P\cL P)^2$ is relatively $\cY$-bounded, then there is a $c_2$ such that
		\begin{equation*}
			\begin{aligned}
				\|\left(Pe^{\frac{\cL}{n}}P\right)^nPx-e^{P\cL P}Px\|_{\cX}&\leq\frac{c_2(1+e^{t\omega})}{n}\|Px\|_{\cY}\,.
			\end{aligned}
		\end{equation*}
	\end{prop}
	This abstract result can be applied to verify the construction of the $X(\theta)$-gate of the bosonic CAT code \cite{Mirrahimi.2014}.
	\begin{ex}\label{ex:bosonic-gate-zeno}
		In 1-mode, we define the Schrödinger CAT-states by 
		\begin{equation*}
			\ket{CAT_{\alpha}^+}\coloneqq \frac{1}{\sqrt{2(1+e^{2|\alpha|^2})}}\left(\ket{\alpha}+\ket{-\alpha}\right)\quad\text{and}\quad\ket{CAT_{\alpha}^-}\coloneqq \frac{1}{\sqrt{2(1-e^{2|\alpha|^2})}}\left(\ket{\alpha}-\ket{-\alpha}\right)\,,
		\end{equation*} 
		which are orthonormal and play the role of the logical qubits in the bosonic CAT code. Then, the projections onto the code space spanned by $\ket{CAT_{\alpha}^+}$ and $\ket{CAT_{\alpha}^-}$ is denoted by $P_\alpha=\ket{CAT^+_{\alpha}}\bra{CAT^+_{\alpha}}+\ket{CAT^-_{\alpha}}\bra{CAT^-_{\alpha}}$. Moreover, the logical $X$ gate is  $X_\alpha=\ket{CAT^+_{\alpha}}\bra{CAT^-_{\alpha}}+\ket{CAT^-_{\alpha}}\bra{CAT^+_{\alpha}}$ so that the rotation around the $X$-axis is described by
		\begin{equation*}
			X(\theta)=\cos(\frac{\theta}{2})(P^+_{\alpha}+P^-_{\alpha})+i\sin(\frac{\theta}{2})X_{\alpha}\,.
		\end{equation*}
		Due to the following identity for $H=a+\ad$
		\begin{equation}
			P_\alpha HP_\alpha=(\alpha+\alpha^\dagger) X_{\alpha}\,,
		\end{equation}
		the above rotation can be achieved by the Zeno effect with $\cP_\alpha(\cdot)=P_\alpha\cdot P_\alpha$
		\begin{equation*}
			\|\left(\cP_\alpha e^{-i\frac{t}{n}[H,\cdot]}\cP_\alpha\right)^nx-e^{-it[P_\alpha HP_{\alpha},\cdot ]}\cP_\alpha x\|_{1\rightarrow1}=\cO\Bigl(\frac{t^2}{n}\Bigr)
		\end{equation*}
		because 
		\begin{equation*}
			e^{tiP_\alpha HP_\alpha}P_\alpha=\cos(t(\alpha+\alpha^\dagger))P_\alpha+i\sin(t(\alpha+\alpha^\dagger))X_{\alpha}\,.
		\end{equation*} 
	\end{ex}
	Finally, we extend the result to more general contractions $M$ as well as contractive evolution systems similarly to the case of bounded generators considered in \cite{Mobus.2019}. This extends the result in \cite{Moebus.2023} to evolution systems.
	\begin{thm*}[See Theorem \ref{thm:general-zeno}]
		Let $I_T\ni(t,s)\mapsto V(t,s)\in\cB(\cX)$ be a $C_0$-evolution system with generator family $(\cL_t,\cD(\cL_t))_{t\in[0,T]}$, $M\in\cB(\cX)$ a contraction, and $P$ a projection satisfying 
		\begin{equation*}
			\norm{M^n-P}_{\cX\rightarrow\cX}\leq\delta^n
		\end{equation*}
		for $\delta\in(0,1)$ and all $n\in\N$.  Moreover, we assume the asymptotic Zeno condition:
		\begin{equation*}
			\norm{PV(t,s)(\1-P)}_{\cX\rightarrow\cX}\leq (t-s)b\quad\text{and}\quad\norm{(\1-P) V(t,s)P}_{\cX\rightarrow\cX}\leq (t-s)b
		\end{equation*}
		for $b\geq0$ and all $t\in[0,T]$. Furthermore, there is a Banach space $(\cY\subset\cX,\|\cdot\|_{\cY})$ invariant under $P$ so that the closure of $(P\cL_t P,\cY)_{t\in[0,T]}$ generates a unique contractive $C_0$-evolution system $T(t,s)$ for $(t,s)\in I_T$ commuting with $P$. Moreover, assume that $\cL_tP$, $P\cL_tP$, and $P\cL_tP\cL_sP$ are well-defined on $\cY$ and relatively $\cY$-bounded for all $(t,s)\in I_T$, and $(\cY,\|\cdot\|_{\cY})$ is a $P\cL_t P$-stable subspace, in particular $\|T(t,s)\|_{\cY\rightarrow\cY}\leq \tilde{c}e^{(t-s)\omega}$. Then, there is a $c\geq0$ so that
		\begin{equation*}
			\begin{aligned}
				\Bigl\|\prod_{j=1}^n MV(s_{j},s_{j-1})x-T(t,0)Px\Bigr\|_{\cX}\leq c\frac{t^2}{n}\left(\|x\|_{\cX}+e^{t\omega}\|Px\|_{\cY}\right)
			\end{aligned}
		\end{equation*}
		for all $t\geq0$, $s_0=0<\frac{t}{n}<...<\frac{t(n-1)}{n}<t=s_n$, $x\in\cY$ and $n\in\N$.
	\end{thm*}

\section{\bfseries Trotter-like product formulas}\label{sec:trotter}
	In this section, we discuss quantitative bounds for Trotter-like product formulas. We begin with the fundamental Trotter product formula \cite{Trotter.1959} and enhance the convergence rate through the Suzuki approach \cite{Suzuki.1997}. Additionally, we generalize the involved semigroups to evolution systems. Finally, we delve into bosonic examples, in particular form bosonic Hamiltonian learning, as well as error correction. Most of our proofs rely on a telescopic sum approach known for example from proving the Lie-Trotter product formula \cite{Reed.1980,Zagrebnov.2019} and generalized in the following.
	\begin{lem}\label{prop:telescopic-sum}
		Let $I_T\ni(t,s)\mapsto F(t,s)\in\cB(\cX)$ be a vector-valued map with $\|F(t,s)\|_{\cX\rightarrow\cX}\leq1$ and $F(s,s)=\1$ for all $(t,s)\in I_{T}\coloneqq\{(t,s)\,|\,T\geq t\geq s\geq0\}$. Moreover, let $T(t,s)$ be an evolution system, then for all $(t,s)\in I_{T}$
		\begin{equation*}
			\|\prod_{j=1}^n F(s_{j},s_{j-1})x-T(t,s)x\|_{\cX}\leq n\max_{j\in\{1,...,n\}}\left\{\|\left(F(s_{j},s_{j-1})-T(s_{j},s_{j-1})\right)T(s_{j-1},s_{0})x\|_{\cX}\right\}\,,
		\end{equation*}
		for any increasing sequence $s=s_0<s_1<...<s_{n-1}<s_n=t$. Note that the product over $F(s_j,s_{j-1})$ is in decreasing order from the left to the right.
	\end{lem}
	\begin{proof}
		The statement directly follows by the following telescopic sum \cite{Reed.1980,Zagrebnov.2019}:
		\begin{equation}\label{eq:telescopic-sum}
			\begin{aligned}
				\prod_{j=1}^n F_jx-\prod_{j=1}^{n} T_jx&=\sum_{k=1}^n\left(\prod_{j=k+1}^{n} F_j\right)(F_{k}-T_{k})\left(\prod_{j=1}^{k-1} T_j\right)x\,,
			\end{aligned}
		\end{equation}
		where the products are in decreasing order from the left to the right. Then, the definition $T_j=T(s_j,s_{j-1})$, the semigroup property, i.e.
		\begin{equation*}
			T(t,s)=T(s_n,s_{n-1})\circ T(s_{n-1},s_{n-2})\circ\cdots\circ T(s_{1},s_{0})\eqqcolon \prod_{j=1}^n T_j\,,
		\end{equation*}
		and the contractivity of $F(t,s)$ combined with the submultiplicativity of the norm finish the proof.
	\end{proof}
	\subsection{Trotter product formula}\label{subsec:trotter-product}
		In the following, we consider a set of assumptions on the generators $(\cL,\cD(\cL))$ and $(\cK,\cD(\cK))$ such that a quantitative error bound on Equation \ref{eq:lie-trotter} can be found.
		\begin{thm}\label{thm:trotter-tele}
			Let $(\cL,\cD(\cL))$ and $(\cK,\cD(\cK))$ be generators of two contractive $C_0$-semigroups and $(\cY\subset\cX,\|\cdot\|_{\cY})$ a Banach space such that $\cY\subset\cD(\cL+\cK)$ and\,\footnote{In the following, we omit the assumption $\cY\subset\cD(\cL+\cK)$ because it is clear by the notation $(\cL+\cK,\cY)$} the closure of $(\cL+\cK,\cY)$ generates a contractive $C_0$-semigroup on $\cX$. If additionally $\cK$, $\cL$, $\cK^2$, $\cL\cK$, $\cK\cL$, $\cL^2$ are relatively $\cY$-bounded (see Eq.~(\ref{eq:relative-Y-bounded-op})) and  $(\cY,\|\cdot\|_{\cY})$ is $(\cL+\cK)$-stable, in particular $\|e^{t(\cL+\cK)}\|_{\cY\rightarrow\cY}\leq \tilde{c}e^{t\omega}$, then there is a $c\geq0$ such that for all $x\in\cY$, $t\geq0$, and $n\in\N$
			\begin{equation*}
				\begin{aligned}
					\norm{\left(e^{\frac{t}{n}\cL}e^{\frac{t}{n}\cK}\right)^nx-e^{t(\cL+\cK)}x}_{\cX}&\leq c\frac{t^2}{n}e^{t\omega}\|x\|_{\cY}\,.
				\end{aligned}
			\end{equation*}
		\end{thm}
		\begin{proof}
			By rescaling the generators, assume $0\leq s\leq t\leq 1$. Then, we apply Proposition \ref{prop:telescopic-sum} to 
			\begin{equation*}
				F(t,s)=e^{(t-s)\cL}e^{(t-s)\cK}\quad T(t,s)=e^{(t-s)(\cL+\cK)}
			\end{equation*}
			w.r.t.~the sequence $s_0=0<\frac{1}{n}<...<\frac{n-1}{n}<1$ so that
			\begin{equation*}
				\|\left(e^{\frac{\cL}{n}}e^{\frac{\cK}{n}}\right)^nx-e^{\cL+\cK}x\|_{\cX}\leq n \max_{s\in[0,1-1/n]}\|\left(e^{\frac{\cL}{n}}e^{\frac{\cK}{n}}-e^{\frac{\cL+\cK}{n}}\right)e^{s(\cL+\cK)}x\|_{\cX}\,.
			\end{equation*}
			Then, the integral equation (\ref{eq:fund-thm-calculus}) shows for $x_s\coloneqq e^{s(\cL+\cK)}x$
			\begin{equation*}
				\begin{aligned}
					e^{\frac{\cL}{n}}e^{\frac{\cK}{n}}x_s-e^{\frac{\cL+\cK}{n}}x_s&=(e^{\frac{\cL}{n}}-\1)x_s+e^{\frac{\cL}{n}}(e^{\frac{\cK}{n}}-\1)x_s-(e^{\frac{\cL+\cK}{n}}-\1)x_s\\
					&=\frac{1}{n}\int_0^1\left(e^{s_1\frac{\cL}{n}}\cL+e^{\frac{\cL}{n}}e^{s_1\frac{\cK}{n}}\cK-e^{s_1\frac{\cL+\cK}{n}}(\cL+\cK)\right)x_sds_1\\
					&=\frac{1}{n^2}\int_0^1\int_0^1\Bigl(s_1e^{s_1s_2\frac{\cL}{n}}\cL^2+e^{s_2\frac{\cL}{n}}\cL\cK+s_1e^{\frac{\cL}{n}}e^{s_1s_2\frac{\cK}{n}}\cK^2\\
					&\qquad\qquad\qquad-s_1e^{s_1s_2\frac{\cL+\cK}{n}}(\cL+\cK)^2\Bigr)x_sds_2ds_1\,.
				\end{aligned}
			\end{equation*}
			The integral on the right-hand side is well-defined due to the continuity of the integrand in $s_1$ and $s_2$ given by the assumptions on the domain, i.e.~$x_s\in\cY$. Moreover, the assumption that $\cL^2$ is $\cY$-bounded and $\cY$ $(\cK+\cL)$-admissible induces constants $a_{\cL^2},\tilde{c},\tilde{\omega}\geq0$ satisfying
			\begin{equation*}
				\|e^{s\frac{\cL}{n}}\cL^2x_s\|_{\cX}\leq\|\cL^2x_s\|_{\cX}\leq a_{\cL^2}\|x_s\|_{\cY}\leq a_{\cL^2}\tilde{c}e^{s\tilde{\omega}}\|x\|_{\cY}\,.
			\end{equation*}
			The same calculation holds for all other terms of the integrand in the last line with appropriate constants $a_{\cL\cK},a_{\cK^2},a_{\cK\cL}\geq0$ defined by the $\cY$-boundedness. This implies  
			\begin{equation*}
				\begin{aligned}
					\|\left(e^{\frac{\cL}{n}}e^{\frac{\cK}{n}}\right)^nx-e^{\cL+\cK}x\|_{\cX}&\leq\frac{1}{n}\biggl(a_{\cL^2}+\frac{3}{2}a_{\cL\cK}+a_{\cK^2}+a_{\cK\cL}\biggr) \tilde{c}e^{\tilde{\omega}}\|x\|_{\cY}\,.
				\end{aligned}
			\end{equation*}
			and finishes the proof of the theorem by $t^2c=\bigl(a_{\cL^2}+\frac{3}{2}a_{\cL\cK}+a_{\cK^2}+a_{\cK\cL}\bigr) \tilde{c}$ and $t\omega =\tilde{\omega}$. Here, $\cL$, $\cK$ were rescaled to $t\cL$, $t\cK$ so that $a_{A}$ becomes $t^2a_{A}$ for all $A\in\{\cL^2, \cK^2, \cK\cL, \cL\cK\}$. 
		\end{proof}
		\begin{rmk*}
			Note that in the above proof the assumption that $(\cY,\|\cdot\|_{\cY})$ is a $(\cK+\cL$)-stable Banach space can be weakened to a normed space invariant under the $C_0$-semigroup $e^{t(\cL+\cK)}$ satisfying $\|e^{t(\cL+\cK)}\|_{\cY\rightarrow\cY}\leq \tilde{c}e^{t\omega}$. Together with the other assumptions, the same poof works.
		\end{rmk*}
		
	\subsection{Trotter-Suzuki product formula}\label{subsec:trotter-suzuki}
		Reconsidering the proof strategy of Theorem \ref{thm:trotter-tele}, improving the convergence rate of the product formula reduces to finding an appropriate $F(t)$ so that
		\begin{equation}\label{eq:trotter-suzuki}
			\|F\Bigl(\frac{1}{n}\Bigr)-e^{\frac{\cL+\cK}{n}}\|_{\cX\rightarrow\cX}=\cO\Bigl(\frac{1}{n^{m+1}}\Bigr)\,.
		\end{equation}
	 	In the case of the symmetrized Lie product formula $m=2$ and $F(t)\coloneqq e^{t\frac{\cL}{2}} e^{t\cK} e^{t\frac{\cL}{2}}$ (see Thm.~\ref{thm:symmetrized-Lie}). Following this idea, we assume the structure
	 	\begin{equation*}
	 		F(t)\coloneqq e^{tp_1A}e^{tp_2B}e^{tp_3A}\cdots e^{tp_{m-1}A}e^{tp_mB}
	 	\end{equation*}
	 	for $p_1,...,p_m\in\R$. For that, \citeauthor{Suzuki.1976} proved different schemes to achieve higher order convergence rates (see \cite{Hatano.2005} for an overview) for bounded generators. In the following, we consider a general framework and explicitly demonstrate the symmetrized product formula afterwards.
	 	\begin{lem}\label{lem:general-suzuki}
	 		Let $(\cL,\cD(\cL))$ and $(\cK,\cD(\cK))$ be generators of two contractive $C_0$-semigroups, $m\in\N_{\geq1}$, and $(\cY\subset\cX,\|\cdot\|_{\cY})$ a Banach space such that the closure of $(\cL+\cK,\cY)$ generates a $C_0$-semigroup on $\cX$. Moreover, we assume there exists a contractive and strongly continuous map $t\mapsto F(t)$ for $t\geq0$ such that  
	 		\begin{equation}\label{eq:assumption-suzuki-boundedness}
	 			\|F(t)-e^{t(\cL+\cK)}\|_{\cY\rightarrow\cX}\leq b t^{m+1}\,.
	 		\end{equation}
	 		If $(\cY,\|\cdot\|_{\cY})$ is $(\cL+\cK)$-stable, in particular $\|e^{t(\cL+\cK)}\|_{\cY\rightarrow\cY}\leq \tilde{c}e^{t\omega}$, then for all $x\in\cY$ and $t\geq0$
	 		\begin{equation*}
	 			\norm{\left(F\Bigl(\frac{t}{n}\Bigr)\right)^nx-e^{t(\cL+\cK)}x}_{\cX}\leq \frac{t^{m+1}}{n^m}\tilde{c}be^{t\omega}\|x\|_{\cY}\,.
	 		\end{equation*}
	 	\end{lem}
 		\begin{proof}
 			Similar to the proof of Theorem \ref{thm:trotter-tele}, we apply Proposition \ref{prop:telescopic-sum} to 
 			\begin{equation*}
 				F(t,s)=F(t-s)\quad T(t,0s=e^{(t-s)(\cL+\cK)}
 			\end{equation*}
 			w.r.t.~the sequence $s_0=0<t\frac{1}{n}<...<t\frac{n-1}{n}<s_n=t$ so that
 			\begin{equation*}
 				\begin{aligned}
 					\|\left(F\Bigl(\frac{t}{n}\Bigr)\right)^nx-e^{t(\cL+\cK)}x\|_{\cX}&\leq n \max_{s\in[0,t(1-1/n)]}\|\left(F\Bigl(\frac{t}{n}\Bigr)-e^{\frac{t}{n}(\cL+\cK)}\right)e^{s(\cL+\cK)}x\|_{\cX}\\
 					&\leq \frac{t^{m+1}}{n^m}\tilde{c}be^{t\omega}\|x\|_{\cY}
 				\end{aligned}
 			\end{equation*}
 			where we used the assumption in Equation \ref{eq:assumption-suzuki-boundedness} and that $\cY$ is $({\cL+\cK})$-stable in the last step.
 		\end{proof}
 		On quantum lattice system in the infinite volume limit, \cite{Bachmann.2022} shows that Equation \ref{eq:assumption-suzuki-boundedness} is satisfied and the above convergence rate in the strong topology is achieved. Next, we apply the above structure to the symmetrized Trotter product formula (or second order Trotter-Suzuki product formula) and give assumptions such that quadratic convergence is achieved. 
		\begin{thm}\label{thm:symmetrized-Lie}
			Let $(\cL,\cD(\cL))$ and $(\cK,\cD(\cK))$ be generators of two contractive $C_0$-semigroups and $(\cY\subset\cX,\|\cdot\|_{\cY})$ a Banach space such that the closure of $(\cY,\cL+\cK)$ generates a contractive $C_0$-semigroup in $\cX$. If additionally all products of $\cK$ and $\cL$ up to the third power\footnote{$\cL$, $\cK$, $\cL^2$, $\cK^2$, $\cL\cK$, $\cK\cL$, $\cK^3$, $\cK^2\cL$, $\cK\cL^2$, $\cL^3$, $\cL^2\cK$, $\cK\cL\cK$, $\cL \cK^2$, $\cL\cK\cL$} are relatively $\cY$-bounded and $(\cY,\|\cdot\|_{\cY})$ is $(\cL+\cK)$-stable, in particular $\|e^{t(\cL+\cK)}\|_{\cY\rightarrow\cY}\leq \tilde{c}e^{t\omega}$, then for all $x\in\cY$ and $t\geq0$ there is a $c\geq0$ so that
			\begin{equation*}
				\begin{aligned}
					\norm{\left(e^{\frac{t}{2n}\cK}e^{\frac{t}{n}\cL}e^{\frac{t}{2n}\cK}\right)^nx-e^{t(\cL+\cK)}x}_{\cX}&\leq c\frac{t^3}{n^2}e^{t\omega}\|x\|_{\cY}\,.
				\end{aligned}
			\end{equation*}
		\end{thm}
		\begin{proof}
			Similarly to the proof of Theorem \ref{thm:trotter-tele}, we use the integral equation (\ref{eq:fund-thm-calculus}) to prove a bound like Equation \ref{eq:assumption-suzuki-boundedness} to apply Lemma \ref{lem:general-suzuki}. For $x\in\cY$
            \begin{equation*}
                \begin{aligned}
                    &e^{\frac{\cK}{2n}}e^{\frac{\cL}{n}}e^{\frac{\cK}{2n}}x-e^{\frac{\cL+\cK}{n}}x\\
                    &=\frac{1}{n}\int_0^1e^{s_1\frac{\cK}{2n}}\frac{\cK}{2}x+e^{\frac{\cK}{2n}}e^{s_1\frac{\cL}{n}}\cL x+e^{\frac{\cK}{2n}}e^{\frac{\cL}{n}}e^{s_1\frac{\cK}{2n}}\frac{\cK}{2}x-e^{s_1\frac{\cL+\cK}{n}}(\cL+\cK)xds_1\\
                    &=\frac{1}{n^2}\int_0^1\int_0^1s_1e^{s_1s_2\frac{\cK}{2n}}\frac{\cK^2}{4}x+e^{s_2\frac{\cK}{2n}}\frac{\cK}{2}\cL x+s_1e^{\frac{\cK}{2n}}e^{s_1s_2\frac{\cL}{n}}\cL^2 x-s_1e^{s_1s_2\frac{\cL+\cK}{n}}(\cL+\cK)^2x\\
                    &\qquad\qquad +e^{s_2\frac{\cK}{2n}}\frac{\cK^2}{4}x+e^{\frac{\cK}{2n}}e^{s_2\frac{\cL}{n}}\cL\frac{\cK}{2} x+s_1e^{\frac{\cK}{2n}}e^{\frac{\cL}{n}}e^{s_1s_2\frac{\cK}{2n}}\frac{\cK^2}{4}x ds_2ds_1\\
                    &=\frac{1}{n^3}\int_0^1\int_0^1\int_0^1s_1^2s_2e^{s_1s_2s_3\frac{\cK}{2n}}\frac{\cK^3}{8}x+s_2e^{s_2s_3\frac{\cK}{2n}}\frac{\cK^2}{4}\cL x+s_1e^{s_3\frac{\cK}{2n}}\frac{\cK}{2}\cL^2 x+s_1^2s_2e^{\frac{\cK}{2n}}e^{s_1s_2s_3\frac{\cL}{n}}\cL^3 x\\
                    &\qquad\qquad +s_2e^{s_2s_3\frac{\cK}{2n}}\frac{\cK^3}{8}x+e^{s_3\frac{\cK}{2n}}\frac{\cK}{2}\cL\frac{\cK}{2} x+s_2e^{\frac{\cK}{2n}}e^{s_2s_3\frac{\cL}{n}}\cL^2\frac{\cK}{2} x\\
                    &\qquad\qquad +s_1e^{s_3\frac{\cK}{2n}}\frac{\cK^3}{8}x+s_1e^{\frac{\cK}{2n}}e^{s_3\frac{\cL}{n}}\cL\frac{\cK^2}{4}x+s_1^2s_2e^{\frac{\cK}{2n}}e^{\frac{\cL}{n}}e^{s_1s_2s_3\frac{\cK}{2n}}\frac{\cK^3}{8}x\\
                    &\qquad\qquad-s_1^2s_2e^{s_1s_2s_3\frac{\cL+\cK}{n}}(\cL+\cK)^3x\, ds_3ds_2ds_1\,.
                \end{aligned}
			\end{equation*}
			Note that the input $x\in\cY$ plays the role of $x_s=e^{s(\cL+\cK)}x$ in the proof of Theorem \ref{thm:trotter-tele}. Again, the integral is well-defined (see \cite[Lem.~II.1.3]{Engel.2000}) and the relative boundedness
			\begin{equation*}
				\|A\|_{\cY\rightarrow\cX}\leq a_{A}
			\end{equation*}
			for $A\in\{\cK^3,\cK^2\cL,\cK\cL^2,\cL^3,\cK\cL\cK,\cL\cK\cL,\cL^2\cK,\cL\cK^2\}$ shows
			\begin{equation*}
				\begin{aligned}
					\|e^{\frac{\cK}{2n}}&e^{\frac{\cL}{n}}e^{\frac{\cK}{2n}}x-e^{\frac{\cL+\cK}{n}}x\|_{\cX}\\
					&\leq\frac{1}{n^3}\int_0^1\int_0^1\int_0^1\biggl(\frac{2s_1^2s_2+s_2+s_1}{8}\|\cK^3x\|_{\cX}+\frac{s_2}{4}\|\cK^2\cL x\|_{\cX}+\frac{s_1}{2}\|\cK\cL^2 x\|_{\cX}+s_1^2s_2\|\cL^3 x\|_{\cX}\\
					&\qquad\qquad +\frac{1}{4}\|\cK\cL\cK x\|_{\cX}+\frac{s_2}{2}\|\cL^2\cK x\|_{\cX}+\frac{s_1}{4}\|\cL\cK^2x\|_{\cX}+s_1^2s_2\|(\cL+\cK)^3x\|_{\cX}\biggr)\, ds_3ds_2ds_1\\
					&\leq\frac{1}{6n^3}\biggl(a_{\cK^3}+\frac{3}{4}a_{\cK^2\cL}+\frac{3}{2}a_{\cK\cL^2}+a_{\cL^3}+\frac{3}{2}a_{\cK\cL\cK}+\frac{3}{2}a_{\cL^2\cK}+\frac{3}{4}a_{\cL\cK^2}\biggr)\|x\|_{\cY}\\
					&\quad+\frac{1}{6n^3}\biggl(a_{\cL^3}+a_{\cL^2\cK}+a_{\cL\cK\cL}+a_{\cL\cK^2}+a_{\cK\cL^2}+a_{\cK\cL\cK}+a_{\cK^2\cL}+a_{\cK^3}\biggr)\|x\|_{\cY}\\
					&\leq\frac{1}{6n^3}\biggl(2a_{\cK^3}+\frac{7}{4}a_{\cK^2\cL}+\frac{5}{2}a_{\cK\cL^2}+2a_{\cL^3}+\frac{5}{2}a_{\cK\cL\cK}+\frac{5}{2}a_{\cL^2\cK}+\frac{7}{4}a_{\cL\cK^2}\biggr)\|x\|_{\cY}\,.
				\end{aligned}
			\end{equation*}
			By rescaling the time and the above calculation, there is a $b\geq0$
			\begin{equation*}
				\|e^{\frac{t}{2n}\cK}e^{\frac{t}{n}\cL}e^{\frac{t}{2n}\cK}x-e^{\frac{t}{n}(\cL+\cK)}x\|_{\cX}\leq b\frac{t^3}{n^3}\|x\|_{\cY}
			\end{equation*}
			so that Lemma \ref{lem:general-suzuki} finishes the proof of the statement. 
		\end{proof}

	\subsection{Time-dependent Trotter product formula}\label{subsec:time-dep-trotter}
		Finally, we prove a Trotter product formula for evolution systems of the form
		\begin{equation*}
			\prod_{j=1}^n U(s_{j},s_{j-1})V(s_{j},s_{j-1})\rightarrow T(t,s)\qquad\text{for}\qquad n\rightarrow\infty\,,
		\end{equation*}
		where $0\leq s=s_0<...<s_n=t$ is a partition of $[s,t]$, $U(t,s)$ and $V(t,s)$ are two differentiable $C_0$-evolution systems with generator families $(\cL_s,\cD(\cL_s))$ and $(\cK_s,\cD(\cK_s))$, and appropriate assumptions.
		\begin{thm}\label{thm:trotter-time-indep}
			Let $U(t,s)$ and $V(t,s)$ be contractive $C_0$-evolution system for $(t,s)\in I_T$ and with generator families $(\cL_t,\cD(\cL_t))_{t\in[0,T]}$ and $(\cK_t,\cD(\cK_t))_{t\in[0,T]}$. Moreover, assume that there is a Banach space $(\cY\subset\cX,\|\cdot\|_{\cY})$ so that the closure of $(\cL_t+\cK_t,\cY)_{t\in[0,T]}$ generates a unique contractive $C_0$-evolution systems $T(t,s)$ on $\cX$ for $(t,s)\in I_T$. If $\cL_s$, $\cK_s$, $\cL_s\cL_t$, $\cK_s\cL_t$, $\cL_s\cK_t$, $\cK_s\cK_t$ are well-defined on $\cY$, relatively $\cY$-bounded, and $(\cY,\|\cdot\|_{\cY})$ is $(\cL_t+\cK_t)$-stable uniformly for all $(s,t)\in I_T$, in particular $\|T(t,s)\|_{\cY\rightarrow\cY}\leq \tilde{c}e^{(t-s)\omega}$, then there is a $c\geq0$ so that
			\begin{equation*}
				\begin{aligned}
					\norm{\biggl(\prod_{j=1}^n U(s_{j},s_{j-1})V(s_{j},s_{j-1})\biggr)x-T(t,0)x}_{\cX}&\leq c\frac{t^2}{n}e^{t\omega}\|x\|_{\cY}
				\end{aligned}
			\end{equation*}
			for all $x\in\cY$, $t\in[0,T]$, and $s_0=0<\frac{t}{n}<...<\frac{(n-1)t}{n}<t=s_n$.
		\end{thm}
		\begin{proof}
			First, we rescale the time parameter so that $0\leq s\leq t\leq1$ and apply Proposition \ref{prop:telescopic-sum} to
			\begin{equation*}
				F(t,s)=U(t,s)V(t,s)
			\end{equation*}
			with respect to the sequence $s_0=0<\frac{1}{n}<...<\frac{n-1}{n}<1=s_n$. Then, we achieve
			\begin{equation*}
				\begin{aligned}
					\|\biggl(\prod_{j=1}^nU(s_{j}&,s_{j-1})V(s_{j},s_{j-1})\biggr)x-T(1,0)x\|_{\cX}\\
					&\leq n \max_{j\in\{1,...,n\}}\|\Bigl(U(s_{j},s_{j-1})V(s_{j},s_{j-1})-T(s_{j},s_{j-1})\Bigr)T(s_{j-1},s_0)x\|_{\cX}\,.
				\end{aligned}
			\end{equation*}
			Then, the integral equation (\ref{eq:fund-thm-calculus}) as well as Lemma \ref{lem:product-continuity} show with $s_j(\tau)=s_j+\frac{1-\tau}{n}$
			\begin{equation}\label{eq:time-trotter-op-bound}
				\begin{aligned}
					&U(s_{j},s_{j-1})V(s_{j},s_{j-1})\tilde{x}-T(s_{j},s_{j-1})\tilde{x}\\
					&=(U(s_{j},s_{j-1})-\1)\tilde{x}+U(s_{j},s_{j-1})(V(s_{j},s_{j-1})-\1)\tilde{x}-(T(s_{j},s_{j-1})-\1)\tilde{x}\\
					&=\frac{1}{n}\int_{0}^{1}\frac{\partial}{\partial \tau_1}\Bigl(U(s_{j},s_{j-1}(\tau_1))\tilde{x}+U(s_{j},s_{j-1})V(s_{j},s_{j-1}(\tau_1))\tilde{x}-T(s_{j},s_{j-1}(\tau_1))\tilde{x}\Bigr)d\tau_1\\
					&=\frac{1}{n}\int_{0}^{1}\biggl(U(s_{j},s_{j-1}(\tau_1))\cL_{s_{j-1}(\tau_1)}\tilde{x}+U(s_{j},s_{j-1})V(s_{j},s_{j-1}(\tau_1))\cK_{s_{j-1}(\tau_1)}\tilde{x}\\
					&\qquad\qquad\quad-T(s_{j},s_{j-1}(\tau_1))(\cL_{s_{j-1}(\tau_1)}+\cK_{s_{j-1}(\tau_1)})\tilde{x}\Bigr)d\tau_1\\
					&=\frac{1}{n^2}\iint_{0}^{1}\Bigl(\tau_1U(s_{j},s_{j-1}(\tau_1\tau_2))\cL_{s_{j-1}(\tau_1\tau_2)}\cL_{s_{j-1}(\tau_1)}\tilde{x}+U(s_{j},s_{j-1}(\tau_2))\cL_{s_{j-1}(\tau_2)}\cK_{s_{j-1}(\tau_1)}\tilde{x}\\
					&\qquad\qquad\quad+\tau_1U(s_{j},s_{j-1})V(s_{j},s_{j-1}(\tau_1\tau_2))\cK_{s_{j-1}(\tau_1\tau_2)}\cK_{s_{j-1}(\tau_1)}\tilde{x}\\
					&\qquad\qquad\quad-\tau_1T(s_{j},s_{j-1}(\tau_1\tau_2))(\cL_{s_{j-1}(\tau_1\tau_2)}+\cK_{s_{j-1}(\tau_1\tau_2)})(\cL_{s_{j-1}(\tau_1)}+\cK_{s_{j-1}(\tau_1)})\tilde{x}\biggr)d\tau_2d\tau_1\,,
				\end{aligned}
			\end{equation}
			where $\tilde{x}\coloneqq T(s_{j-1},s_0)x$. Next, we use the relative boundedness assumptions as well as the property that $(\cY,\|\cdot\|_{\cY})$ defines an $(\cL_s+\cK_s)$-stable subspace to show
			\begin{equation*}
				\begin{aligned}
					\|\cL_{s_{j-1}(\tau_1)}\tilde{x}\|_{\cX}&\leq b\|T(s_{j-1},s_0)x\|_{\cY}\leq b\,\tilde{c}\,e^{(s_{j-1}-s_0)\omega}\|x\|_{\cY}
				\end{aligned}
			\end{equation*}
			and
			\begin{equation*}
				\begin{aligned}
					\|U(s_{j},s_{j-1}(\tau_1\tau_2))\cL_{s_{j-1}(\tau_1\tau_2)}\cL_{s_{j-1}(\tau_1)}\tilde{x}\|_{\cX}&\leq\|\cL_{s_{j-1}(\tau_1\tau_2)}\cL_{s_{j-1}(\tau_1)}T(s_{j-1},s_0)x\|_{\cX}\\
					&\leq b\|T(s_{j-1},s_0)x\|_{\cY}\\
					&\leq b\,\tilde{c}\,e^{(s_{j-1}-s_0)\omega}\|x\|_{\cY}\\
				\end{aligned}
			\end{equation*}
			as well as similar bounds for all the other terms coming up in Equation \ref{eq:time-trotter-op-bound}. These two bounds finish the proof of the error bound of the statement.
		\end{proof}
		
	\subsection{Application: Bosonic strongly continuous semigroups}\label{subsec:bosonic-trotter}
		In this section, we focus on bosonic $C_0$-semigroups defined on the space of self-adjoint trace-class operators on the Fock space (see Sec.~\ref{subsec:bosonic-semigroups}). Specially, we consider generators admitting the following GKSL structure
		\begin{equation}\label{eq:lindblad2}
			\cL:\cT_f \to \cT_f \quad x \mapsto\cL(x) = - i [H, x] + \sum\limits_{j = 1}^K L_j x L_j^\dagger  - \frac{1}{2}\{L_j^\dagger L_j, x\}
		\end{equation}
		with $K\in\mathbb{N}$, $L_j\coloneqq p_j(a_1,\ad_1,...)$ for all $j\in\{1,...,m\}$, and a symmetric $H\coloneqq p_H(a_1,\ad_1,...)$. Note that in the time-dependent case the coefficients would additionally depend continuously on time. Assuming Equation \ref{eq:assumsobolevstability}, i.e.~there are sequences $(\k_r)_{r \in \N}\subset\N^m$ diverging coordinate-wise to infinity and $(\omega_{\k_r})_{r\in\N}$ so that 
		\begin{equation*}
			\tr[\cW^{\k_r}\cL(x)] \le \omega_{\k_r} \tr[\cW^{\k_r}(x)]
		\end{equation*}
		for all positive semi-definite $x \in \cT_f$, not only shows that the closure of $(\cL,\cT_f)$ generates a quantum dynamical $C_0$-semigroup, but also that it is Sobolev preserving, i.e.
		\begin{equation}\label{eq:sobolev-preserving}
			\|T_t\|_{W^{\k_r, 1} \to W^{\k_r, 1}} \le e^{\omega_{\k_r} t}
		\end{equation}
		for all $t\geq0$. Moreover, the above inequality holds not only for the points of $\k_r$ but can be extended by interpolation (for more details \cite[Thm.~4.4, Lem.~E.4]{Gondolf.2023}). Note that $W^{\k_r, 1}$ is a $\cL$-admissible subspace and plays exactly the role of the Banach space $\cY$ discussed in the statements before. One interesting class of examples for Sobolev-preserving semigroups is this of decoupled quadratic Hamiltonians with an additional polynomial in the local number operators, i.e.
		\begin{equation}\label{eq:quadratic-hamiltonian}
			H=\sum_{i\in\{1,..,m\}}\sum_{k_i+\ell_i\leq2} \lambda_{k_i\ell_i}(\ad_i)^{k_i}a_i^{\ell_i}+p(N_1,...,N_m)
		\end{equation}
		defined on $\cH_f$ and with $\lambda_{k_i\ell_i}=\overline{\lambda}_{\ell_ik_i}\in\C$ as well as $p\in\R[x_1,...,x_m]$. This example can be eassily extended to $C_0$-evolution system by assuming that $\lambda_{k_i\ell_i}$ depends continuously on time for $t\in[0,T]$. Then, the operator family $(-i[H(t),\cdot],\cT_f)_{t\in[0,T]}$ generates a unique $C_0$-evolution system (see proof of\cite[Lem.~5.6]{Gondolf.2023} and \cite[Thm.~4.4-4.5]{Gondolf.2023}). 
  
        Another example is the generator of the Ornstein-Uhlenbeck semigroup which is defined as
        \begin{equation}
	        \cL_{\operatorname{qOU}} = \lambda^2 \cL[a] + \mu^2 \cL[\ad]
        \end{equation}
        for $\mu, \lambda \ge0$. In \cite[Sec.~5.1]{Gondolf.2023}, we showed that this semigroup is Sobolev-preserving as well. 
        
        Generalizing the first term in the above generator, one achieves the so called $l$-photon driven dissipation generated by 
        \begin{equation}\label{eq:l-photon-dissipation}
	           \cL[a^l-\alpha^l]
        \end{equation}
        for $\alpha\in\C$ and $l\in\N$. This operator is used in bosonic error correction codes and satisfies the moment stability bound (\ref{eq:assumsobolevstability}) as well. Therefore, it generates a Sobolev preserving quantum Markov semigroup (see \cite[Sec.~5.2]{Gondolf.2023}). 
        
        Besides the property of the semigroups to be Sobolev-preserving, the finite degree assumption on $\cL$ in Equation \ref{eq:lindblad2} implies by \cite[Lem.~E3]{Gondolf.2023} the existence of a $\k\in\N^m$ and $b\geq0$ such that
		\begin{equation}\label{eq:relative-boundedness}
			\|\cL(x)\|_1\leq b\|x\|_{W^{\k,1}}\,.
		\end{equation}
		With the help of these two properties, we can apply the above theorems to bosonic Sobolev preserving $C_0$-semigroups. For that, let $(\cL,\cD(\cL))$ and $(\cK,\cD(\cK))$ be generators of contractive $C_0$-semigroups on $\cT_{1,\operatorname{sa}}$ satisfying Equation \ref{eq:lindblad2}. Additionally, assume that the closure of $(\cL+\cK,\cT_f)$ generates a Sobolev preserving, which satisfies Equation \ref{eq:sobolev-preserving} in particular. This follows directly from \cite{Gondolf.2023} if Equation \ref{eq:assumsobolevstability} is satisfied. Then, there is a $\k$ and $b\geq0$ such that Theorem \ref{thm:trotter-time-indep} and \ref{thm:symmetrized-Lie} show
		\begin{equation*}
			\begin{aligned}
				\norm{\left(e^{\frac{t}{n}\cL}e^{\frac{t}{n}\cK}\right)^nx-e^{t(\cL+\cK)}x}_{\cX}&\leq b\frac{t^2}{n}e^{t\omega_{2\k}}\|x\|_{\cW^{2\k,1}}\\
				\norm{\left(e^{\frac{t}{2n}\cK}e^{\frac{t}{n}\cL}e^{\frac{t}{2n}\cK}\right)^nx-e^{t(\cL+\cK)}x}_{\cX}&\leq b\frac{t^3}{n^2}e^{t\omega_{3\k}}\|x\|_{\cW^{3\k,1}}\,.
			\end{aligned}
		\end{equation*}
		Here $\cK$ and $\cL$ can be for example quadratic Hamiltonian of the form Equation \ref{eq:quadratic-hamiltonian} or one of the other examples mentioned above.
		The constant $\k$ coming up in the above bound can be understood as the component-wise maximum of $\k$'s satisfying Equation \ref{eq:relative-boundedness} for $\cL$ and $\cK$. Then, concatenations of $\cL$ and $\cK$ simply accumulate in the regularity $\k$ of the Sobolev space --- we get $2\k$ or $3\k$.
		
		An interesting observation is that only the limiting semigroup generated by $\overline{\cL+\cK}$ needs to be Sobolev preserving, which enables the scheme to be used as a "regularization" method in the sense of Sobolev preserving semigroups. More explicitly, many semigroups even satisfy a stronger assumption than Equation \ref{eq:assumsobolevstability}, which is 
		\begin{equation}\label{eq:assumsobolevstability-stronger}
			\tr[\cW^{\k_r}\cL(x)] \le -\mu_{\k_r} \tr[\cW^{\k_r+\l}(x)]+c_{\k_r}
		\end{equation}
		for a sequence $\k_r$ converging component-wise to infinity, $\l\in\N^m$ depending highly on the structure of the generator $\cL$, and constants $\mu_{\k_r},c_{\k_r}\geq0$. This stronger assumption not only allows us to improve the bound Equation \ref{eq:sobolev-preserving} for states to a constant scaling in $t$ (see \cite[Prop.~5.1]{Gondolf.2023}), but also to bound all other terms which are not of leading order by a constant. For example, we can show that for any Hamiltonian defined by $H=p_{H}(a_1,\ad_1,...,a_m,\ad_m)$ there is a $l\in\N$ so that
		\begin{equation}\label{eq:l-photon-dissipation-hamiltonian}
			-i[H,\cdot]+\cL[a^l-\alpha^l](\cdot)
		\end{equation}
		generates a Sobolev preserving semigroup and, in particular, it satisfies Equation \ref{eq:assumsobolevstability-stronger}. This idea is for example used in the recent bosonic Hamiltonian Learning scheme \cite{Moebus.2023Learning}, in which a similar method is used to derive Lieb-Robinson bounds.  
		
		The same reasoning translates to the time-dependent case as shortly mentioned before in the case of decoupled quadratic Hamiltonians. Here, the coefficients of the polynomials in Equation \ref{eq:lindblad2} depend continuously on time $t\in[0,T]$. Then, a similar statement as Theorem \ref{thm:multimode-generation-theorem} can be proven if Equation \ref{eq:assumsobolevstability} is satisfied with a time-independent right-hand side (see \cite[Thm.~4.5]{Gondolf.2023}). Assuming that all involved evolution systems are well-defined and of GKSL form (\ref{eq:lindblad2}), the assumption that the Trotter-limit is again Sobolev preserving implies 
		\begin{equation*}
			\begin{aligned}
				\norm{\prod_{j=1}^n U(s_{j+1},s_j)V(s_{j+1},s_j)x-T(t,s)x}_{\cX}\leq c\frac{t^2}{n}e^{t\omega}\|x\|_{W^{2\k,1}}
			\end{aligned}
		\end{equation*}
		for a $c,\omega\geq0$, $\k\in\N^m$ and all $t\in[0,T]$, $x\in W^{2\k,1}$  by Theorem \ref{thm:trotter-time-indep}. Trotter's product formula for time-independent and time-dependent quantum Markov semigroups play a key role in the simulation of certain Hamiltonian. For example, it was shown that local Hamiltonians can be simulated efficiently \cite{Lloyd.1996} by using a Trotter product formula. As mentioned in the introduction, the idea is to decompose a Hamiltonian or Lindbladian in simpler for example local terms. It is important to mention that it is not clear yet how to utilize information propagation bounds to achieve a bound dependent of the system size (see \cite{Moebus.2023Learning}). In the case of evolutions systems, the question of efficient simulation is directly related to adiabatic quantum computation proposed by \citeauthor{Feynman.1982} \cite{Feynman.1982} and developed by \cite{Lloyd.1996} and many other works like \cite{Farhi.2000, Farhi.2001}. In bosonic systems, similar protocols are constructed and the present work can be utilized to find explicit error bounds on the Trotter product formula (see for example \cite{Sun.2020}). 
		
\section{\bfseries Quantum Zeno effect}\label{sec:zeno}
	As mentioned, the quantum Zeno effect is an operator product formula of a contractive projection $P\in\cB(\cX)$ and a $C_0$-semigroup $e^{t\cL}$:
	\begin{equation*}
		(Pe^{\frac{\cL}{n}})^n\rightarrow e^{P\cL P}P\qquad\text{for}\quad n\rightarrow\infty\,.
	\end{equation*}
	We start with the projective Zeno effect without assuming the asymptotic Zeno condition \cite[Eq.~6]{Moebus.2023}. Then, we extend the result to more general contractions $M$ as in \cite{Moebus.2023,Becker.2021}. Finally, we apply the abstract theory again to Sobolev preserving semigroups, in particular, to an example from bosonic continuous error correction. 
	\subsection{Projective Zeno effect}\label{subsec:projective-zeno}
		We start with the simplest case presented in Equation \ref{eq:zeno} where a $C_0$-semigroup is interrupted by a contractive projection $P$.
		\begin{prop}\label{prop:simple-zeno}
			Let $(\cL,\cD(\cL))$ be a generator of a contractive $C_0$-semigroup and $P$ a projective contraction. Assume there is a Banach space $(\cY,\|\cdot\|_{\cY})$ invariant under $P$ such that the closure of $(P\cL P,\cY)$ defines a contractive $C_0$-semigroup, $\cL P$ and $\cL^2P$ are relatively $\cY$-bounded, and $(\cY,\|\cdot\|_{\cY})$ is a $P\cL P$-stable subspace, in particular $\|e^{tP\cL P}\|_{\cY\rightarrow\cY}\leq \tilde{c}e^{t\omega}$. Then, there is a $c_1\geq0$ such that for all $t\geq0$, $x\in\cY$, and $n\in\N$
			\begin{equation*}
				\begin{aligned}
					\|\left(Pe^{t\frac{\cL}{n}}P\right)^nPx-e^{tP\cL P}Px\|_{\cX}&\leq\frac{c_1t^2}{n}\left(e^{t\omega}\|Px\|_{\cY}+\|(P\cL P)^2x\|_{\cX}\right).
				\end{aligned}
			\end{equation*}
			If we additionally assume that $(P\cL P)^2$ is relatively $\cY$-bounded, then there is a $c_2$ such that
			\begin{equation*}
				\begin{aligned}
					\|\left(Pe^{\frac{\cL}{n}}P\right)^nPx-e^{P\cL P}Px\|_{\cX}&\leq\frac{c_2(1+e^{t\omega})}{n}\|Px\|_{\cY}\,.
				\end{aligned}
			\end{equation*}
		\end{prop}
		\begin{proof}
			By rescaling the generator $\cL$, we assume $t\in[0,1]$. Then, we apply Proposition \ref{prop:telescopic-sum} to 
			\begin{equation*}
				F(t,s)=Pe^{(t-s)\cL}P\quad T(t,s)=e^{(t-s)P\cL P}
			\end{equation*}
			defined on the space $P\cX$ on which $P$ is the identity operator and with respect to the sequence $s_0=0<\frac{1}{n}<...<\frac{n-1}{n}<1=s_n$ so that
			\begin{equation*}
				\|\left(Pe^{\frac{\cL}{n}}P\right)^nPx-e^{P\cL P}Px\|_{\cX}\leq n \max_{s\in[0,1-1/n]}\|\left(Pe^{\frac{\cL}{n}}P-e^{\frac{P\cL P}{n}}\right)e^{sP\cL P}Px\|_{\cX}\,.
			\end{equation*}
			Then, we use $e^{P\cL P}Px=Pe^{P\cL P}Px$ and the integral equation (\ref{eq:fund-thm-calculus}) shows
			\begin{equation*}
				\begin{aligned}
					Pe^{\frac{\cL}{n}}Px_s-e^{\frac{P\cL P}{n}}Px_s&=\frac{1}{n^2}\int_0^1\int_0^1\left(s_1Pe^{s_1s_2\frac{\cL}{n}}\cL^2P-s_1e^{s_1s_2P\frac{P\cL P}{n}}(P\cL P)^2\right)x_sds_2ds_1\,,
				\end{aligned}
			\end{equation*}
			where $x_s\coloneqq e^{sP\cL P}Px$. In a first step, we use the contractivity of the involved semigroups and of $P$ to show
			\begin{equation*}
				\begin{aligned}
					\|\int_0^1\int_0^1s_1Pe^{s_1s_2\frac{\cL}{n}}\cL^2Px_sds_2ds_1\|_{\cX}&\leq\int_0^1\int_0^1s_1\|\cL^2Px_s\|_{\cX}ds_2ds_1=\frac{1}{2}\|\cL^2Px_s\|_{\cX}\,.
				\end{aligned}
			\end{equation*}
			Then, the assumptions that $\cL^2P$ is relatively $\cY$-bounded, i.e.~$\|\cL^2P\|_{\cY\rightarrow\cX}<a_1$ for an $a_1\geq0$ and $\cY$ is a $P\cL P$-stable subspace, in particular $\|e^{t\cL}\|_{\cY\rightarrow\cY}\leq b_1e^{t\omega}$ for a $b_1,\omega\geq0$ are utilized to show
			\begin{equation}\label{eq:simple-zeno-rel-boundedness}
				\begin{aligned}
					\|\int_0^1\int_0^1s_1Pe^{s_1s_2\frac{\cL}{n}}\cL^2Px_sds_2ds_1\|_{\cX}&\leq\frac{1}{2}\|\cL^2Px_s\|_{\cX}\leq\frac{a_1}{2}\|e^{sP\cL P}Px\|_{\cY}\leq\frac{a_1b_1}{2}e^{\omega}\|Px\|_{\cY}\,.
				\end{aligned}
			\end{equation}
			Therefore, $c_1=\frac{a_1b_1}{2}$ and rescaling of $t$ implies
			\begin{equation*}
				\begin{aligned}
					\|\left(Pe^{t\frac{\cL}{n}}P\right)^nPx-e^{tP\cL P}Px\|_{\cX}&\leq\frac{c_1t^2}{n}\left(e^{t\omega}\|Px\|_{\cY}+\|(P\cL P)^2x\|_{\cX}\right).
				\end{aligned}
			\end{equation*}
			If $(P\cL P)^2$ is additionally relatively $\cY$-bounded with constants $a_2$, we directly achieve 
			\begin{equation*}
				\begin{aligned}
					\|\left(Pe^{\frac{\cL}{n}}P\right)^nPx-e^{P\cL P}Px\|_{\cX}&\leq\frac{c_2(1+e^{t\omega})}{n}\|Px\|_{\cY}\,,
				\end{aligned}
			\end{equation*}
			with $c_2=\max\{c_1,\frac{a_2}{2}\}$. This finishes the proof of the theorem.
		\end{proof}
		\begin{rmk}\label{rmk:possible-space-choice-zeno}
			To get a better understanding of the above result, we compare it to Theorem I in \cite{Moebus.2023}. Assume that $(P\cL P,\cD(P\cL P))$ generates a contractive $C_0$-semigroup. Then, $\cY=\cD((P\cL P)^2)\cap \cD(P\cL P)$ defines a Banach space with the norm $\|x\|_{\cY}\coloneqq \|x\|_{\cX}+\|P \cL Px\|_{\cX}+\|(P \cL P)^2x\|_{\cX}$. Clearly, $\cY$ is invariant under $P$ and an even $P\cL P$-admissible subspace, in particular, $\|e^{tP\cL P}\|_{\cX\rightarrow\cY}\leq 1$. Next, the asymptotic Zeno condition could be used for the following relative bound
			\begin{equation*}
				\begin{aligned}
					\frac{1}{n^2}\|\int_0^1\int_0^1&s_1Pe^{s_1s_2\frac{\cL}{n}}\cL^2Px_sds_2ds_1\|_{\cX}\\
					&=\frac{1}{n}\|\int_0^1Pe^{s_1\frac{\cL}{n}}(P+\1-P)\cL Px_sds_1-P\cL Px_s\|_{\cX}\\
					&\leq\frac{1}{n}\|\int_0^1Pe^{s_1\frac{\cL}{n}}(\1-P)\cL Px_sds_1+\frac{1}{n}Pe^{s_1\frac{\cL}{n}}\cL P\cL Px_sds_1\|_{\cX}\\
					&\leq\frac{1}{n^2}\left(b^2\|x_s\|_{\cX}+b\|P\cL Px_s\|_{\cX}+\|(P\cL P)^2x_s\|_{\cX}\right)\\
					&\leq\frac{1}{n^2}\max\{b^2,b,1\}\|x_s\|_{\cY}\,,
				\end{aligned}
			\end{equation*}
			which is slightly different from the relative bound used in Equation \ref{eq:simple-zeno-rel-boundedness}. Here, $b$ is defined in Theorem I in \cite{Moebus.2023}.
		\end{rmk}
		In the next step, we generalize the above result to $C_0$-evolution systems:
		\begin{thm}\label{thm:generell-projective-zeno}
			Let $I_T\ni(t,s)\mapsto V(t,s)\in\cB(\cX)$ be a $C_0$-evolution system with generator family $(\cL_t,\cD(\cL_t))_{t\in[0,T]}$ and $P\in\cB(\cX)$ a projective contraction. Assume that there is a Banach space $(\cY\subset\cX,\|\cdot\|_{\cY})$ invariant under $P$ so that the closure of $(P\cL_t P,\cY)_{t\in[0,T]}$ generates a unique contractive $C_0$-evolution systems $T(t,s)$ for $(t,s)\in I_T$ commuting with $P$. Moreover, assume that $\cL_tP$, $\cL_t\cL_sP$, $P\cL_tP$, and $P\cL_tP\cL_sP$ are well-defined on $\cY$ and relatively $\cY$-bounded uniformly for all $(t,s)\in I_T$. If $(\cY,\|\cdot\|_{\cY})$ is a $P\cL_t P$-stable subspace, in particular $\|T(t,s)\|_{\cY\rightarrow\cY}\leq \tilde{c}e^{(t-s)\omega}$, then there is a $c\geq0$ so that
			\begin{equation*}
				\begin{aligned}
					\Bigl\|\prod_{j=1}^n PV(s_{j},s_{j-1})Px-T(t,0)Px\Bigr\|_{\cX}\leq c\frac{t^2}{n}e^{t\omega}\|Px\|_{\cY}
				\end{aligned}
			\end{equation*}
			for all $t\geq0$, $s_0=0<\frac{t}{n}<...<\frac{t(n-1)}{n}<t=s_n$, $x\in\cY$, and $n\in\N$.
		\end{thm}
		Note that the assumption about the existence of the Banach space $\cY$ is motivated by hyperbolic evolution systems for which by definition an admissible subspace exists (see Thm.~\ref{thm:time-dependent-semigroups}).
		\begin{proof}
			As before, we rescale $t\in[0,1]$ and apply Proposition \ref{prop:telescopic-sum} to 
			\begin{equation*}
				F(t,s)=PV(t,s)P
			\end{equation*}
			defined on $P\cX$ w.r.t.~the sequence $s_0=0<\frac{1}{n}<...<\frac{n-1}{n}<1=s_n$ so that
			\begin{equation*}
				\begin{aligned}
					\Bigl\|\prod_{j=1}^n PV(s_{j},s_{j-1})&Px-T(t,0)Px\Bigr\|_{\cX}\\
					&\leq n \max_{j\in\{1,...,n\}}\|\left(PV(s_{j},s_{j-1})P-T(s_j,s_{j-1})\right)T(s_{j-1},s_0)Px\|_{\cX}\,.
				\end{aligned}
			\end{equation*}
			Then, we use $T(s_{j-1},s_0)Px=PT(s_{j-1},s_0)Px$ and the integral equation (\ref{eq:fund-thm-calculus}) to show with $s_j(\tau)=s_j+\frac{1-\tau_1}{n}$
			\begin{equation*}
				\begin{aligned}
					P&V(s_{j},s_{j-1})Px_j-T(s_j,s_{j-1})Px_j\\
					&=\frac{1}{n}\int_{0}^{1}\Biggl(PV(s_{j},s_{j-1}(\tau_1))\cL_{s_{j-1}(\tau_1)}P-T(s_{j},s_{j-1}(\tau_1))P\cL_{s_{j-1}(\tau_1)}P\Biggr)x_jd\tau_1\\
					&=\frac{1}{n^2}\int_{0}^{1}\int_{0}^{1}\Biggl(\tau_1PV(s_{j},s_{j-1}(\tau_1\tau_2))\cL_{s_{j-1}(\tau_1\tau_2)}\cL_{s_{j-1}(\tau_1)}P\\
					&\qquad\qquad\qquad-\tau_1T(s_{j},s_{j-1}(\tau_1\tau_2))P\cL_{s_{j-1}(\tau_1\tau_2)}P\cL_{s_{j-1}(\tau_1)}P\Biggr)x_jd\tau_2d\tau_1\,,
				\end{aligned}
			\end{equation*}
			where $x_j\coloneqq PT(s_{j-1},s_0)Px$. Note that the integral above is again well-defined by Lemma \ref{lem:product-continuity} as well as \cite[Lem.~B.15]{Engel.2000}. As in the proof of Proposition \ref{prop:simple-zeno}, we use the assumption to show
			\begin{equation*}
				\begin{aligned}
					\|\int_{0}^{1}\int_{0}^{1}\tau_1PV(s_{j},s_{j-1}(\tau_1))\cL_{s_{j-1}(\tau_1\tau_2)}\cL_{s_{j-1}(\tau_1)}Px_jd\tau_2d\tau_1\|_{\cX}&\leq\frac{a_1}{2}\|T(s_{j-1},s_0)Px\|_{\cY}\\
					&\leq c_1e^{\omega}\|Px\|_{\cY}
				\end{aligned}
			\end{equation*}
			for $c_1=\frac{a_1,b_1}{2}$ with $a_1,b_1\geq0$ and similarly  
			\begin{equation*}
				\begin{aligned}
					\|\int_{0}^{1}\int_{0}^{1}\tau_1T(s_{j},s_{j-1}(\tau_1))P\cL_{s_{j-1}(\tau_1\tau_2)}P\cL_{s_{j-1}(\tau_1)}Px_jd\tau_2d\tau_1\|_{\cX}&\leq c_2e^{\omega}\|x_j\|_{\cY}
				\end{aligned}
			\end{equation*}
			for $c_2=\frac{a_2,b_2}{2}$ with $a_2,b_2\geq0$. Rescaling and combining both bounds show
			\begin{equation*}
				\begin{aligned}
					\Bigl\|\prod_{j=1}^n PV(s_{j},s_{j-1})&Px-T(t,0)Px\Bigr\|_{\cX}\leq\frac{ct^2}{n}e^{t\omega}\|Px\|_{\cY}
				\end{aligned}
			\end{equation*}
			with $c=2\max\{c_1,c_2\}$.
		\end{proof}
		
	\subsection{Time-dependent Zeno with general measurement}\label{subsec:general-zeno-measurement}
		In the next step, we generalize Theorem I in \cite{Moebus.2023} to $C_0$-evolution systems. For that, we use again the help of a reference subspace $\cY$:
		\begin{thm}\label{thm:general-zeno}
			Let $I_T\ni(t,s)\mapsto V(t,s)\in\cB(\cX)$ be a $C_0$-evolution system with generator family $(\cL_t,\cD(\cL_t))_{t\in[0,T]}$, $M\in\cB(\cX)$ a contraction, and $P$ a projection satisfying 
			\begin{equation}\label{unifpower1}
				\norm{M^n-P}_{\cX\rightarrow\cX}\leq\delta^n
			\end{equation}
			for $\delta\in(0,1)$ and all $n\in\N$.  Moreover, we assume the asymptotic Zeno condition:
			\begin{equation}\label{appx-eq:thm1-asympzeno}
				\norm{PV(t,s)(\1-P)}_{\cX\rightarrow\cX}\leq (t-s)b\quad\text{and}\quad\norm{(\1-P) V(t,s)P}_{\cX\rightarrow\cX}\leq (t-s)b
			\end{equation}
			for $b\geq0$ and all $t\in[0,T]$. Furthermore, there is a Banach space $(\cY\subset\cX,\|\cdot\|_{\cY})$ invariant under $P$ so that the closure of $(P\cL_t P,\cY)_{t\in[0,T]}$ generates a unique contractive $C_0$-evolution system $T(t,s)$ for $(t,s)\in I_T$ commuting with $P$. Moreover, assume that $\cL_tP$, $P\cL_tP$, and $P\cL_tP\cL_sP$ are well-defined on $\cY$ and relatively $\cY$-bounded for all $(t,s)\in I_T$, and $(\cY,\|\cdot\|_{\cY})$ is a $P\cL_t P$-stable subspace, in particular $\|T(t,s)\|_{\cY\rightarrow\cY}\leq \tilde{c}e^{(t-s)\omega}$. Then, there are $c_1,c_2\geq0$ so that
			\begin{equation*}
				\begin{aligned}
					\Bigl\|\prod_{j=1}^n MV(s_{j},s_{j-1})x-T(t,0)Px\Bigr\|_{\cX}\leq \left(\delta^n+\frac{c_1(t+t^2)}{n}\right)\|x\|_{\cX}+c_2\frac{t^2}{n}e^{t\omega}\|Px\|_{\cY}
				\end{aligned}
			\end{equation*}
			for all $t\geq0$, $s_0=0<\frac{t}{n}<...<\frac{t(n-1)}{n}<t=s_n$, $x\in\cY$ and $n\in\N$.
		\end{thm}
		\begin{proof}
			As in the proof of \cite[Thm.~5.1]{Moebus.2023}, we approximate $M$ by $P$ in the first step. To do so, we follow the proof of Lemma 5.2 in \cite{Moebus.2023}. For that 
			\begin{equation*}
				\begin{aligned}
					\Bigl\|\prod_{j=1}^n &(P+(\1-P)M)V(s_{j},s_{j-1})x-\prod_{j=1}^n PV(s_{j},s_{j-1})x\Bigr\|_{\cX}\\
					&=\Bigl\|\sum_{\tau\in\{0,1\}^n\backslash(1,...,1)}\prod_{j=1}^n (\tau_jP+(1-\tau_j)(\1-P)M)V(s_{j},s_{j-1})x\Bigr\|_{\cX}\,.
				\end{aligned}
			\end{equation*}
			Next, we analyse how many transitions in terms of $PV(s_{j},s_{j-1})(1-P)$ exist in the above sum. This is specified in Lemma 5.4 in \cite{Moebus.2023}. Next, we use the bounds
			\begin{equation*}
				\|M(\1-P)V(s_{j},s_{j-1})\|_{\cX\rightarrow\cX}\leq\delta\qquad\text{and}\qquad\|M(\1-P)V(s_{j},s_{j-1})PMV(s_{j+1},s_{j})\|\leq\delta\frac{b}{n}
			\end{equation*}
			and follow the calculations on \cite[p.~12]{Moebus.2023} to show
			\begin{equation}\label{eq:general-zeno-step1}
				\norm{\prod_{j=1}^n MV(s_{j},s_{j-1})x-\prod_{j=1}^n PV(s_{j},s_{j-1})Px}_{\cX}\leq\left(\delta^n+\frac{tb}{n}+\frac{1}{n}\frac{tb(2+tb)(\delta-\delta^{n})}{1-\delta}e^{2b}\right)\norm{x}_{\cX}
			\end{equation}
			for all $x\in\cX$. Then, we follow the proof of Theorem \ref{thm:generell-projective-zeno} and apply Proposition \ref{prop:telescopic-sum} to 
			\begin{equation*}
				F(t,s)=PV(t,s)P
			\end{equation*}
			defined on $P\cX$ w.r.t.~the sequence $s_0=0<\frac{t}{n}<...<\frac{t(n-1)}{n}<t=s_n$ so that
			\begin{equation*}
				\begin{aligned}
					\Bigl\|\prod_{j=1}^n PV(s_{j},s_{j-1})&Px-T(t,0)Px\Bigr\|_{\cX}\\
					&\leq n \max_{j\in\{1,...,n\}}\|\left(PV(s_{j},s_{j-1})P-T(s_j,s_{j-1})\right)T(s_{j-1},s_0)Px\|_{\cX}\,.
				\end{aligned}
			\end{equation*}
			Then, we use $T(s_{j-1},s_0)Px=PT(s_{j-1},s_0)Px$ and the integral equation (\ref{eq:fund-thm-calculus}) to show with $s_j(\tau)=s_j+\frac{(1-\tau_1)t}{n}$
			\begin{equation*}
				\begin{aligned}
					P&V(s_{j},s_{j-1})Px_j-T(s_j,s_{j-1})Px_j\\
					&=\frac{t}{n}\int_{0}^{1}\Biggl(PV(s_{j},s_{j-1}(\tau_1))\cL_{s_{j-1}(\tau_1)}P-T(s_{j},s_{j-1}(\tau_1))P\cL_{s_{j-1}(\tau_1)}P\Biggr)x_jd\tau_1\\
					&=\frac{t}{n}\int_{0}^{1}\Biggl(PV(s_{j},s_{j-1}(\tau_1))(\1-P)\cL_{s_{j-1}(\tau_1)}P\\
					&\qquad\qquad\qquad+PV(s_{j},s_{j-1}(\tau_1))P\cL_{s_{j-1}(\tau_1)}P-T(s_{j},s_{j-1}(\tau_1))P\cL_{s_{j-1}(\tau_1)}P\Biggr)x_jd\tau_1
				\end{aligned}
			\end{equation*}
			and for the last two terms
			\begin{equation*}
				\begin{aligned}
					\frac{t}{n}\int_{0}^{1}\Biggl(&PV(s_{j},s_{j-1}(\tau_1))P\cL_{s_{j-1}(\tau_1)}P-T(s_{j},s_{j-1}(\tau_1))P\cL_{s_{j-1}(\tau_1)}P\Biggr)x_jd\tau_1\\
					&=\frac{t^2}{n^2}\int_{0}^{1}\int_{0}^{1}\Biggl(\tau_1PV(s_{j},s_{j-1}(\tau_1\tau_2))(P+\1-P)\cL_{s_{j-1}(\tau_1\tau_2)}P\cL_{s_{j-1}(\tau_1)}P\\
					&\qquad\qquad\qquad-\tau_1T(s_{j},s_{j-1}(\tau_1\tau_2))P\cL_{s_{j-1}(\tau_1\tau_2)}P\cL_{s_{j-1}(\tau_1)}P\Biggr)x_jd\tau_2d\tau_1\,,
				\end{aligned}
			\end{equation*}
			where $x_j\coloneqq PT(s_{j-1},s_0)Px$. Note that the integral above is again well-defined by Lemma \ref{lem:product-continuity} or \cite[Lem.~B.15]{Engel.2000}. Next, we use the relative $\cY$-boundedness as in Proposition \ref{prop:simple-zeno}, but also the asymptotic Zeno condition, which implies boundedness for all $t\in[0,T]$
			\begin{equation*}
				\|P\cL_t(\1-P)\|_{\cX\rightarrow\cX},\|(\1-P)\cL_tP\|_{\cX\rightarrow\cX}\leq b\,.
			\end{equation*}
			Then, $\|P\cL_tPx\|_{\cX}\leq a_1\|x\|_{\cY}$ and $\|P\cL_tP\cL_sPx\|_{\cX}\leq a_2\|x\|_{\cY}$ for all $(t,s)\in I_T$ show
			\begin{equation*}
				\begin{aligned}
					\|PV(s_{j},s_{j-1})Px_j-T(s_j,s_{j-1})Px_j\|_{\cX}&\leq\frac{t^2}{n^2}\left(b^2\|x_j\|_{\cX}+\left(\frac{a_1b}{2}+a_2\right)\|x_j\|_{\cY}\right)\,.
				\end{aligned}
			\end{equation*}
			Applying $\|PT(s_{j-1},s_0)Px\|_{\cY}\leq a_3e^{\omega}\|Px\|_{\cY}$ to the above bound shows
			\begin{equation*}
				\begin{aligned}
					\|PV(s_{j},s_{j-1})Px_j-T(s_j,s_{j-1})Px_j\|_{\cX}\leq\frac{t^2}{n^2}\left(b^2\|Px\|_{\cX}+\left(\frac{a_1b}{2}+a_2\right)e^{\omega}\|Px\|_{\cY}\right)\,.
				\end{aligned}
			\end{equation*}
			Combining the bound finishes the proof of the result.
		\end{proof}
		\begin{rmk*}
			If we assume $\cL_t\cL_sP$ (instead of $P\cL_tP$) is relatively $\cY$-bounded for all $(t,s)\in I_T$, then, the proof directly follows by Theorem \ref{thm:generell-projective-zeno} after Equation \ref{eq:general-zeno-step1}. However, then Theorem \ref{thm:general-zeno} does not reduce directly to Theorem I in \cite{Moebus.2023}.
		\end{rmk*}
	
	\subsection{Application: The \texorpdfstring{$l$}{l}-photon driven dissipation in error correction}\label{subsec:bosonic-error-correction}
		We continue with the example in Section \ref{subsec:bosonic-trotter} about time-(in)dependent $C_0$-semigroups generated by GKSL-operators (\ref{eq:lindblad2}), in particular, Sobolev preserving semigroups that satisfy Equation \ref{eq:assumsobolevstability} (see Sec.~\ref{subsec:bosonic-trotter} and \ref{subsec:bosonic-semigroups} for more details). In the following, we shortly repeat the proof of these two properties for the $l$-photon driven dissipation for simplicity in one-mode \cite[Lem.~5.3]{Gondolf.2023} and give two specific examples for projections satisfying the assumptions of the statements before (\ref{prop:simple-zeno}, Theorem \ref{thm:generell-projective-zeno}, \ref{thm:general-zeno}). The $l$-photon driven dissipation (\ref{eq:l-photon-dissipation-hamiltonian}) defined for $x\in\cT_f$ by
		\begin{equation}
			\cL[a^l-\alpha^l][x]=(a^l-\alpha^l)x((\ad)^l-\Bar{\alpha}^l)-\frac{1}{2}\left\{((\ad)^l-\Bar{\alpha}^l)(a^l-\alpha^l),x\right\}
		\end{equation}
	 	is the fundamental dynamics for the bosonic error correction code introduced in \cite{Mirrahimi.2014, Guillaud.2019, Guillaud.2023}. Using the CCR and the simple commutation relation $af(N+j\1)=f(N+(j+1)\1)a$ for any real-valued function $f:\R\rightarrow\R$, we have for $f(x)=(x+1)^{k/2} 1_{x\ge -1}$ and $g_l(x)\coloneqq f(x)-f(x-l)$:
	 	\begin{equation*}
	 		\begin{aligned}
	 			\cL[a_1^l-\alpha^l]^\dagger(f(N))&=-(N-(l-1)\1)\cdots Ng_l(N)+\frac{1}{2}\left(\overline{\alpha}^la^lg_l(N)+\alpha^lg_l(N)(\ad)^l\right)\,.
	 		\end{aligned}
	 	\end{equation*}
 		Note that the first term contributes a diagonal and the others, two off-diagonals with respect to the Fock basis. The operator can be decomposed in blocks defined via $a=\frac{h_l(n)g_l(n)}{2}$, $b=\frac{h_l(n+l)g_l(n+l)}{2}$, and $c=\frac{1}{2}\sqrt{h_l(n+l)}\bar{\alpha}^lg_l(n+l)$ by
 		\begin{equation*}
 			\begin{aligned}
 				\frac{1}{2}&\Bigl(-a\ket{n}\bra{n}-b\ket{n+l}\bra{n+l}+c^*\ket{n}\bra{n+l}+c\ket{n+l}\bra{n}\Bigr)\\
 				&\coloneqq\frac{1}{2}\Bigl(-h_l(n)g_l(n)\ket{n}\bra{n}-h_l(n+l)g_l(n+l)\ket{n+l}\bra{n+l}\\
 				&\qquad\qquad\qquad+\sqrt{h_l(n+l)}\left(\bar{\alpha}^lg_l(n+l)\ket{n}\bra{n+l}+\alpha^lg_l(n+l)\ket{n+l}\bra{n}\right)\Bigr)
 			\end{aligned}
 		\end{equation*}
 		for $n\in\N$ and  the real valued function $h_l:\R\rightarrow\R, h_l(x)\coloneqq(x-(l-1)\1)\cdots x1_{x\ge l}$. By considering the eigenvalues of the above operator, it can be upper bounded by
 		\begin{equation*}
 			\begin{aligned}
 				\frac{1}{2}&\Bigl(-a\ket{n}\bra{n}-b\ket{n+l}\bra{n+l}+c^*\ket{n}\bra{n+l}+c\ket{n+l}\bra{n}\Bigr)\\
 				&\leq\frac{1}{4}\left(-a-b+\sqrt{(a-b)^2+4|c|^2}\right)\Bigl(\ket{n}\bra{n}+\ket{n+l}\bra{n+l}\Bigr)\\
 				&\leq\frac{1}{2}\left(-\min\{a,b\}+|c|\right)\Bigl(\ket{n}\bra{n}+\ket{n+l}\bra{n+l}\Bigr)\\
 				&=\frac{1}{2}\left(-h_l(n)g_l(n)+\sqrt{h_l(n+l)}|\alpha|^lg_l(n+l)\right)\Bigl(\ket{n}\bra{n}+\ket{n+l}\bra{n+l}\Bigr)\\
 				&\leq\frac{1}{2}\left(-(n+1)^{k/2-1+l}+c_k^{(l)}\right)\Bigl(\ket{n}\bra{n}+\ket{n+l}\bra{n+l}\Bigr)\,,
 			\end{aligned}
 		\end{equation*}
 		where the constant $c_k^{(l)}\geq0$ exists because the first term is of negative leading order (a more detailed discussion can be found in \cite[Lem.~5.3]{Gondolf.2023}). All together, we achieve
		\begin{equation*}
			\cL[a^l-\alpha^l]^\dagger(f(N))\leq-\frac{l}{2}(N+\1)^{k/2-1+l}+c_k^{(l)}\1\,,
		\end{equation*}
	 	which implies Assumption \ref{eq:assumsobolevstability} and a better Sobolev preserving constant \cite[Prop.~5.1]{Gondolf.2023}:
	 	\begin{equation*}
	 		\begin{aligned}
 				\tr\big[\cL[a^l-\alpha^l ](\rho)(N+\1)^{k/2}\big]&\le -\frac{l}{2} \tr\big[\rho\,(N+\1)^{k/2}\big]+c_k^{(l)}\,.
	 		\end{aligned}
	 	\end{equation*}
	 	However, we are interested in the Sobolev preserving property of $(P\cL P,\cD(P\cL P))$ because this is one way to find the right Banach space $\cY=W^{k,1}$ in the statements above. Here, the idea is to choose $k$ large enough such that the relative boundedness assumption is achieved (for comparison see Equation \ref{eq:relative-boundedness}). However, this strongly depends on $P$. To analyze this, we will consider two cases: First, $P(\cdot)=K\cdot K^\perp$ for a projection $K$ diagonal in the Fock basis. Then, the same reasoning as above proves the well-definedness as well as the Sobolev preserving property of the semigroup defined by $P\cL P$. The same holds true for the time-dependent analogues generated by (Eq.~\ref{eq:time-dep-master-equation}). Another example is the $Z(\theta)$ gate of the bosonic error correction code discussed in \cite{Mirrahimi.2014,Guillaud.2019,Guillaud.2023}. The physical codespace to which the evolution is exponentially converging (see Eq.~\ref{eq:convegrenc-l-photon-dissipation}) and \cite{Azouit.2016}) is defined by 
        \begin{equation*}
            \cC_l\coloneqq\spa\left\{\ketbra{\alpha_1}{\alpha_2}\,:\,\alpha_1,\alpha_2\in\left\{\alpha e^{\frac{i2\pi j}{l}}\,|\,j\in\{0,...,l-1\}\right\}\right\}\,,
        \end{equation*}
        where $\ket{\alpha}$ denotes the coherent state. In particular, we are interested in the case $l=2$. Next, we define the Schrödinger CAT-states, which represent the logical qubits:
        \begin{equation*}
        	\ket{CAT_{\alpha}^+}\coloneqq \frac{1}{\sqrt{2(1+e^{2|\alpha|^2})}}\left(\ket{\alpha}+\ket{-\alpha}\right)\quad\text{and}\quad\ket{CAT_{\alpha}^-}\coloneqq \frac{1}{\sqrt{2(1-e^{2|\alpha|^2})}}\left(\ket{\alpha}-\ket{-\alpha}\right)\,.
        \end{equation*}  
    	Note that the above CAT-states are orthonormal. Then, a possible rotation around the $x$-axis is described in \cite{Mirrahimi.2014} and experimentally realized in \cite{Touzard.2018}. The gate is defined by
    	\begin{equation}\label{eq:rotation}
    		X(\theta)=\cos(\frac{\theta}{2})(P^+_{\alpha}+P^-_{\alpha})+i\sin(\frac{\theta}{2})X_{\alpha}\,,
    	\end{equation}
    	where $P^+=\ket{CAT^+_{\alpha}}\bra{CAT^+_{\alpha}}$, $P^-=\ket{CAT^-_{\alpha}}\bra{CAT^-_{\alpha}}$ are the projection of the CAT-states and $X_{\alpha}=\ket{CAT^+_{\alpha}}\bra{CAT^-_{\alpha}}+\ket{CAT^-_{\alpha}}\bra{CAT^+_{\alpha}}$. Specially, here we discuss a construction of a $X(\theta)$-gate by a discrete projective Zeno effect w.r.t.~the projection $P_\alpha=P_\alpha^++P_\alpha^-$ and the driving Hamiltonian $H=a+\ad$. Then, the Zeno limit describes Equation \ref{eq:rotation}:
        \begin{equation}
            P_\alpha HP_\alpha=(\alpha+\alpha^\dagger) X_{\alpha}\,.
        \end{equation}
    	Since $P_\alpha HP_\alpha$ acts only non-trivially on a 2-dimensional subspace, the Zeno semigroup is 
    	\begin{equation}
    		e^{tiP_\alpha HP_\alpha}P_\alpha=\cos(t(\alpha+\alpha^\dagger))P_\alpha+i\sin(t(\alpha+\alpha^\dagger))X_{\alpha}
    	\end{equation}
    	and constructs the rotation (\ref{eq:rotation}) for appropriate $t$. The convergence 
    	\begin{equation*}
    		\|\left(P_\alpha e^{i\frac{t}{n}[H,\cdot]}\right)^nx-e^{ti[P_\alpha HP_{\alpha},P_{\alpha}\cdot P_{\alpha}]}P_\alpha x\|_{1\rightarrow1}=\cO\Bigl(\frac{t^2}{n}\Bigr)
    	\end{equation*}
    	is a consequence of Proposition \ref{prop:simple-zeno} or a simple case of Theorem \ref{thm:general-zeno}. Here, $\cY$ is just defined as $\cP_\alpha\cT_1$ with $\cP_\alpha=P_\alpha\cdot P_\alpha$. By the boundedness of $P_\alpha H$ and $HP_\alpha$ all assumptions are satisfied and the bound holds true in the uniform topology. The same would hold true for more complex and time-dependent $H$ defined by a polynomial in annihilation and creation operator. The case $l=4$ described in \cite{Mirrahimi.2014} can be equally treated.
    	\begin{rmk}\label{rmk:zeno-improvement}
    		In practice the projection is implemented by the $l$-photon driven dissipation semigroup, which converges in the following sense to $P_\alpha$ \cite[Theorem 2]{Azouit.2016}
    		\begin{equation}\label{eq:convegrenc-l-photon-dissipation}
    			\tr[(a^l-\alpha^l)e^{t\cL_l}(\rho) (a^l-\alpha^l)^\dagger]\leq e^{-l!t}\tr[(a^l-\alpha^l)\rho (a^l-\alpha^l)^\dagger]\,.
    		\end{equation}
    		This convergence does not satisfy Theorem \ref{thm:general-zeno}. Therefore, a future line of research could be to improve Equation \ref{eq:convegrenc-l-photon-dissipation} or Assumption \ref{unifpower1} in Theorem \ref{thm:general-zeno} in the direction of Theorem 2 in \cite{Becker.2021}.
    	\end{rmk}
       
\section{\bfseries Conclusion}\label{sec:conclusion} 
	To summarize, we have seen a range of quantitative error bounds for generalizations of Trotter's product formula and the quantum Zeno effect in the case of $C_0$-semigroups. Our primary contribution lies not only in the extension of the Trotter product formula and the quantum Zeno effect to evolution systems (time-dependent semigroups) but also in providing Suzuki-higher order error bounds for the Trotter-Suzuki product formula in the strong operator topology and a large class of bosonic examples. Moreover, we discussed some examples to explore possible definitions of the abstract reference space $\cY$.
	
	For applications, a desirable line of research would entail a detailed discussion of possible choices of $\cY$ as well as examples which does not satisfy the assumption. For the quantum Zeno effect, a generalization of the convergence assumption of the contraction $M$, as discussed in Remark \ref{rmk:zeno-improvement}, would be fruitful. Here, an interesting interplay of the convergence assumption as well as the regularity of the semigroup quantified by the stable subspace could come up.

\vspace{1ex}
\emph{Acknowledgments:} 
I would like to thank Paul Gondolf, Cambyse Rouzé, Robert Salzmann, Simon Becker, Lauritz van Luijk, Niklas Galke, Stefan Teufel, Marius Lemm, and Yu-Jie Liu for valuable feedback and discussions on the topic. I also acknowledge the support of the Munich Center for Quantum Sciences and Technology and of the Deutsche Forschungsgemeinschaft (DFG, German Research Foundation) – TRR 352 – Project-ID 470903074. Furthermore, this work was supported by the University of Nottingham and the University of Tübingen’s funding as part of the Excellence Strategy of the German Federal and State Governments, in close collaboration with the University of Nottingham.

\setlength{\bibitemsep}{0.05ex}
\printbibliography[heading=bibnumbered]

@article{Azouit.2016,
	title = {Well-posedness and convergence of the Lindblad master equation for a quantum harmonic oscillator with multi-photon drive and damping},
	volume = {22},
	ISSN = {1262-3377},
	doi = {10.1051/cocv/2016050},
	number = {4},
	journal = {ESAIM: Control,  Optimisation and Calculus of Variations},
	publisher = {EDP Sciences},
	author = {Azouit,  Rémi and Sarlette,  Alain and Rouchon,  Pierre},
	year = {2016}
}

@article{Bachmann.2022,
	title = {Trotter Product Formulae for {$*$}-Automorphisms of Quan-tum Lattice Systems},
	volume = {23},
	doi = { 10.1007/s00023-022-01207-8},
	number = {12},
	journal = {Annales Henri Poincaré},
	publisher = {Springer Science and Business Media LLC},
	author = {Bachmann,  Sven and Lange,  Markus},
	year = {2022}
}

@misc{Barankai.2018,
	author = {Barankai, Norbert and Zimbor{\'a}s, Zolt{\'a}n},
	title={Generalized quantum Zeno dynamics and ergodic means}, 
	year={2018},
	doi = {10.48550/arXiv.1811.02509},
	eprint={1811.02509},
	archivePrefix={arXiv},
	primaryClass={math-ph}
}

@misc{Becker.2024,
      title={Convergence rates for the Trotter-Kato splitting}, 
      author={Simon Becker and Niklas Galke and Robert Salzmann and Lauritz van Luijk},
      year={2024},
      doi = {10.48550/arXiv.2407.04045},
      archivePrefix={arXiv},
      primaryClass={math-ph},
      eprint={2407.04045}
}

@article{Becker.2021,
	author = {Becker, Simon and Datta, Nilanjana and Salzmann, Robert},
	year = {2021},
	title = {Quantum Zeno Effect in Open Quantum Systems},
	volume = {22},
	number = {11},
	issn = {1424-0661},
	journal = {Annales Henri Poincar{\'e}},
	doi = {10.1007/s00023-021-01075-8}
}

@article{Beskow.1967,
	author = {Beskow, A. and Nilsson, J.},
	year = {1967},
	title = {Concept of wave function and the irreducible representations of the Poincaré group. II. Unstable systems and the exponential decay law.},
	journal = {Inst. of Theoretical Physics, Goteborg}
}

@misc{Burgarth.2023,
	title={Strong Error Bounds for Trotter {\&} Strang-Splittings and Their Implications for Quantum Chemistry}, 
	author={Daniel Burgarth and Paolo Facchi and Alexander Hahn and Mattias Johnsson and Kazuya Yuasa},
	year={2023},
	archivePrefix={arXiv},
	primaryClass={quant-ph},
    doi = {10.48550/arXiv.2312.08044},
    eprint={2312.08044}
}

@article{Burgarth.2020,
	title={Quantum Zeno Dynamics from General Quantum Operations},
	volume={4},
	doi = {10.22331/q-2020-07-06-289},
	journal={Quantum},
	publisher={Verein zur Forderung des Open Access Publizierens in den Quantenwissenschaften},
	author={Burgarth, Daniel and Facchi, Paolo and Nakazato, Hiromichi and Pascazio, Saverio and Yuasa, Kazuya},
	year={2020},
}

@book{Bratteli.1981,
	author = {Bratteli, Ola and Robinson, Derek W.},
	year = {1981},
	title = {Operator Algebras and Quantum Statistical Mechanics II},
	subtitle = {Equilibrium States Models in Quantum Statistical Mechanics},
	publisher = {Springer},
	address = {Berlin, Heidelberg},
	isbn = {978-3-662-09089-3},
	series = {Texts and Monographs in Physics},
	doi = {10.1007/978-3-662-09089-3}
}

@book{Conway.2007,
	title = {A Course in Functional Analysis},
	ISBN = {978-1-4757-4383-8},
	doi = {10.1007/978-1-4757-4383-8},
	journal = {Graduate Texts in Mathematics},
	publisher = {Springer New York},
	author = {Conway,  John B.},
	year = {2007}
}

@article{Dominy.2013,
	author = {Dominy, Jason M. and Paz-Silva, Gerardo A. and Rezakhani, A. T. and Lidar, Daniel A.},
	year = {2013},
	title = {Analysis of the quantum Zeno effect for quantum control and computation},
	volume = {46},
	number = {7},
	journal = {Journal of Physics A: Mathematical and Theoretical},
	doi = {10.1088/1751-8113/46/7/075306}
}

@book{Engel.2000,
	author = {Engel, Klaus-Jochen and Nagel, Rainer},
	year = {2000},
	title = {One-parameter semigroups for linear evolution equations},
	address = {New York and London},
	volume = {194},
	publisher = {Springer},
	series = {Graduate texts in mathematics},
	doi = {10.1007/b97696},
}

@article{Erez.2004,
	author = {Erez, Noam and Aharonov, Yakir and Reznik, Benni and Vaidman, Lev},
	year = {2004},
	title = {Correcting quantum errors with the Zeno effect},
	volume = {69},
	number = {6},
	issn = {1050-2947},
	journal = {Physical review. A, Atomic, molecular, and optical physics},
	doi = {10.1103/PhysRevA.69.062315}
}

@article{Exner.2021,
	author = {Exner, Pavel and Ichinose, Takashi},
	year = {2021},
	title = {Note on a Product Formula Related to Quantum Zeno Dynamics},
	volume = {22},
	number = {5},
	issn = {1424-0661},
	journal = {Annales Henri Poincar{\'e}},
	doi = {10.1007/s00023-020-01014-z}
}

@article{Facchi.2004,
	author = {Facchi, Paolo and Lidar, Daniel A. and Pascazio, Saverio},
	year = {2004},
	title = {Unification of dynamical decoupling and the quantum Zeno effect},
	volume = {69},
	number = {3},
	issn = {1050-2947},
	journal = {Physical review. A, Atomic, molecular, and optical physics},
	doi = {10.1103/PhysRevA.69.032314}
}

@article{Facchi.2008,
	author = {Facchi, Paolo and Pascazio, Saverio},
	year = {2008},
	title = {Quantum Zeno dynamics: mathematical and physical aspects},
	volume = {41},
	number = {49},
	issn = {1751-8121},
	journal = {Journal of Physics A: Mathematical and Theoretical},
	doi = {10.1088/1751-8113/41/49/493001}
}

@misc{Farhi.2000,
	title={Quantum Computation by Adiabatic Evolution}, 
	author={Edward Farhi and Jeffrey Goldstone and Sam Gutmann and Michael Sipser},
	year={2000},
	eprint={quant-ph/0001106},
	archivePrefix={arXiv},
	primaryClass={quant-ph},
	doi = {10.48550/arXiv.quant-ph/0001106}, 
}

@article{Farhi.2001,
	title = {A Quantum Adiabatic Evolution Algorithm Applied to Random Instances of an NP-Complete Problem},
	volume = {292},
	doi = {10.1126/science.1057726},
	number = {5516},
	journal = {Science},
	publisher = {American Association for the Advancement of Science (AAAS)},
	author = {Farhi,  Edward and Goldstone,  Jeffrey and Gutmann,  Sam and Lapan,  Joshua and Lundgren,  Andrew and Preda,  Daniel},
	year = {2001}
}

@article{Feynman.1982,
	title = {Simulating physics with computers},
	volume = {21},
	doi = {10.1007/bf02650179},
	number = {6–7},
	journal = {International Journal of Theoretical Physics},
	publisher = {Springer Science and Business Media LLC},
	author = {Feynman,  Richard P.},
	year = {1982}
}

@book{Giuseppe.1976,
	title={Hyperbolicity},
	editor={Prato, Giuseppe and Geymonat, Giuseppe},
	isbn={9783642111044},
	series={C.I.M.E. Summer Schools},
	year={1976},
	publisher={Springer}
}

@misc{Gondolf.2023,
	title={Energy preserving evolutions over Bosonic systems}, 
	author={Paul Gondolf and Tim Möbus and Cambyse Rouzé},
	year={2024},
	eprint={2307.13801},
	archivePrefix={arXiv},
	primaryClass={quant-ph},
	doi = {10.48550/arXiv.2307.13801}
}

@article{Guillaud.2023,
	title={Quantum computation with cat qubits},
	ISSN={2590-1990},
	doi = {10.21468/SciPostPhysLectNotes.72},
	journal={SciPost Physics Lecture Notes},
	publisher={Stichting SciPost},
	author={Guillaud, Jérémie and Cohen, Joachim and Mirrahimi, Mazyar},
	year={2023}
}

@article{Guillaud.2019,
    title = {Repetition Cat Qubits for Fault-Tolerant Quan-tum Computation},
    author = {Guillaud, J\'er\'emie and Mirrahimi, Mazyar},
    journal = {Physical Review X},
    volume = {9},
    issue = {4},
    year = {2019},
    publisher = {American Physical Society},
    doi = {10.1103/PhysRevX.9.041053}
}

@inbook{Hatano.2005,
	title = {Finding Exponential Product Formulas of Higher Orders},
	ISBN = {9783540315155},
	doi = {10.1007/11526216_2},
	booktitle = {Lecture Notes in Physics},
	publisher = {Springer Berlin Heidelberg},
	author = {Hatano,  Naomichi and Suzuki,  Masuo},
	year = {2005},
	pages = {37–68}
}

@book{Hille.1996,
	author = {Hille, Einar and Phillips, Ralph S.},
	year = {1996},
	title = {Functional analysis and Semi-Groups},
	edition = {Reviewed and expanded edition},
	volume = {31},
	publisher = {American Mathematical Society},
	isbn = {978-0-8218-3395-7},
	series = {Colloquium Publications},
	publisher = {American Mathematical Society},
	doi = {10.1090/coll/031},
}

@article{Ichinose.1998,
	title={Error estimate in operator norm of exponential product formulas for propagators of parabolic evolution equations},
	author={Ichinose, Takashi and Tamura, Hideo},
	year={1998},
	journal={Osaka Journal of Mathematics},
	volume={35},
	pages={751-770},
	publisher = {Osaka University and Osaka Metropolitan University, Departments of Mathematics}
}

@article{Ikeda.2023,
	title = {Minimum Trotterization Formulas for a Time-Dependent Hamiltonian},
	volume = {7},
	ISSN = {2521-327X},
	doi = {10.22331/q-2023-11-06-1168},
	journal = {Quantum},
	publisher = {Verein zur Forderung des Open Access Publizierens in den Quantenwissenschaften},
	author = {Ikeda,  Tatsuhiko N. and Abrar,  Asir and Chuang,  Isaac L. and Sugiura,  Sho},
	year = {2023},
}

@article{Itano.1990,
	author = {Itano, Wayne M. and Heinzen, Daniel J. and Bollinger, John J. and Wineland, David J.},
	year = {1990},
	title = {Quantum Zeno effect},
	volume = {41},
	number = {5},
	issn = {1050-2947},
	journal = {Physical review. A, Atomic, molecular, and optical physics},
	doi = {10.1103/physreva.41.2295}
}

@book{Kato.1995,
	author = {Kato, Tosio},
	year = {1995},
	title = {Perturbation Theory for Linear Operators},
	edition = {2},
	volume = {132},
	publisher = {{Springer International Publishing}},
	isbn = {9783642662829},
	series = {A Series of Comprehensive Studies in Mathematics},
	doi = {10.1007/978-3-642-66282-9}
}

@article{Kato.1978,
	title={Trotter's product formula for an arbitrary pair of self-adjoint contraction semigroup},
	author={Kato, Tosio},
	journal={Topics in Func. Anal., Adv. Math. Suppl. Studies},
	volume={3},
	pages={185--195},
	year={1978}
}

@article{Kato.1974,
	author = {Tosio Kato},
	title = {On the Trotter-Lie product formula},
	volume = {50},
	journal = {Proceedings of the Japan Academy},
	number = {9},
	publisher = {The Japan Academy},
	pages = {694 -- 698},
	year = {1974},
	doi = {10.3792/pja/1195518790}
}

@book{Kreyszig.1989,
	author = {Kreyszig, Erwin},
	year = {1989},
	title = {Introductory Functional Analysis with Applications},
	address = {New York},
	publisher = {{Wiley classics library}},
	isbn = {0471504599}
}

@article{Lloyd.1996,
	title = {Universal Quantum Simulators},
	volume = {273},
	ISSN = {1095-9203},
	doi = {10.1126/science.273.5278.1073},
	number = {5278},
	journal = {Science},
	publisher = {American Association for the Advancement of Science (AAAS)},
	author = {Lloyd,  Seth},
	year = {1996}
}

@article{Mirrahimi.2014,
	doi = {10.1088/1367-2630/16/4/045014},
	year = 2014,
	publisher = {{IOP} Publishing},
	volume = {16},
	number = {4},
	pages = {045014},
	author = {Mazyar Mirrahimi and Zaki Leghtas and Victor V. Albert and Steven Touzard and Robert J Schoelkopf and Liang Jiang and Michel H Devoret},
	title = {Dynamically protected cat-qubits: a new paradigm for universal quan-tum computation},
	journal = {New Journal of Physics}
}

@article{Misra.1977,
	author = {Misra, Baidyanaith and Sudarshan, George},
	year = {1977},
	title = {The Zeno's paradox in quantum theory},
	volume = {18},
	number = {4},
	issn = {0022-2488},
	journal = {Journal of Mathematical Physics},
	doi = {10.1063/1.523304}
}

@article{Mobus.2019,
	author = {M{\"o}bus, Tim and Wolf, Michael M.},
	year = {2019},
	title = {Quantum Zeno effect generalized},
	volume = {60},
	number = {5},
	issn = {0022-2488},
	journal = {Journal of Mathematical Physics},
	doi = {10.1063/1.5090912}
}

@article{Moebus.2023,
	author = {Tim Möbus and Cambyse Rouzé},
	year = {2023},
	title = {Optimal Convergence Rate in the Quantum Zeno Effect for Open Quantum Systems in Infinite Dimensions},
	journal = {Annales Henri Poincaré},
	doi = {10.1007/s00023-022-01241-6}
}

@misc{Moebus.2023Learning,
	title={Dissipation-enabled bosonic Hamiltonian learning via new information-propagation bounds}, 
	author={Tim Möbus and Andreas Bluhm and Matthias C. Caro and Albert H. Werner and Cambyse Rouzé},
	year={2023},
	eprint={2307.15026},
	archivePrefix={arXiv},
	primaryClass={quant-ph},
	doi = {m87m}
}

@article{Neidhardt.2019,
	title={Trotter Product Formula and Linear Evolution Equations on Hilbert Spaces},
	author={Hagen Neidhardt and Artur Stephan and Valentin Anatol'evich Zagrebnov},
	journal={Analysis and Operator Theory},
	year={2019},
	doi = {10.1007/978-3-030-12661-2_13}
}

@book{Nielsen.2012,
	title = {Quantum Computation and Quantum Information: 10th Anniversary Edition},
	ISBN = {9780511976667},
	doi = {10.1017/cbo9780511976667},
	publisher = {Cambridge University Press},
	author = {Nielsen,  Michael A. and Chuang,  Isaac L.},
	year = {2012}
}

@article{Neidhardt.1999,
	title = {Trotter–Kato Product Formula and Operator-Norm Convergence},
	volume = {205},
	ISSN = {1432-0916},
	doi = {10.1007/s002200050671},
	number = {1},
	journal = {Communications in Mathematical Physics},
	publisher = {Springer Science and Business Media LLC},
	author = {Neidhardt,  Hagen and Zagrebnov,  Valentin A.},
	year = {1999}
}

@article{Neidhardt.1998,
	titel = {On Error Estimates for the Trotter–Kato Product Formula},
	volume = {44},
	ISSN = {0377-9017},
	doi = {10.1023/a:1007494816401},
	number = {3},
	journal = {Letters in Mathematical Physics},
	publisher = {Springer Science and Business Media LLC},
	author = {Neidhardt,  Hagen and Zagrebnov,  Valentin A.},
	year = {1998},
	pages = {169–186}
}

@book{Pazy.1983,
	title = {Semigroups of Linear Operators and Applications to Partial Differential Equations},
	ISBN = {9781461255611},
	ISSN = {0066-5452},
	doi = {10.1007/978-1-4612-5561-1},
	journal = {Applied Mathematical Sciences},
	publisher = {Springer New York},
	author = {Pazy, Amnon},
	year = {1983}
}

@book{Reed.1980,
	author = {Reed, Michael and Simon, Barry},
	year = {1980},
	title = {Functional analysis},
	edition = {Revised and Enlarged Edition},
	volume = {1},
	publisher = {{Academic Press}},
	isbn = {9780080570488},
	series = {Methods of Modern Mathematical Physics}
}

@article{Rogava.1993,
	title = {Error bounds for Trotter-type formulas for self-adjoint operators},
	volume = {27},
	ISSN = {1573-8485},
	doi = {d676nr},
	number = {3},
	journal = {Functional Analysis and Its Applications},
	publisher = {Springer Science and Business Media LLC},
	author = {Rogava,  Jemal},
	year = {1993}
}

@incollection{Schmidt.2004,
	author = {Schmidt, Andreas U.},
	title = {Mathematics of the Quantum Zeno Effect},
	pages = {113--143},
	publisher = {{Nova Science Publishers}},
	isbn = {1-59033-905-3},
	editor = {Benton, Charles V.},
	booktitle = {Mathematical physics research on the leading edge},
	year = {2004},
	doi = {10.48550/arXiv.math-ph/0307044}
}

@book{Simon.2015,
	author = {Simon, Barry},
	year = {2015},
	title = {Operator theory},
	address = {Providence, Rhode Island},
	volume = {4},
	publisher = {{American Mathematical Society}},
	isbn = {9781470411039},
	series = {A comprehensive course in analysis}
}

@misc{Stephan.2023,
	title={Trotter-type formula for operator semigroups on product spaces}, 
	author={Artur Stephan},
	year={2023},
	eprint={2307.00419},
	archivePrefix={arXiv},
	primaryClass={math.FA},
	doi = {10.48550/arXiv.2307.00419}
	
}

@article{Sun.2020,
	title = {Trotterized adiabatic quantum simulation and its application to a simple all-optical system},
	volume = {22},
	doi = {10.1088/1367-2630/ab7a31},
	number = {5},
	journal = {New Journal of Physics},
	publisher = {IOP Publishing},
	author = {Sun,  Yifan and Zhang,  Jun-Yi and Byrd,  Mark S and Wu,  Lian-Ao},
	year = {2020}
}

@article{Suzuki.1997,
	author = {Suzuki, Masuo},
	year = {1997},
	title = {Quantum analysis---Non-commutative differential and integral calculi},
	pages = {339--363},
	volume = {183},
	number = {2},
	issn = {0010-3616},
	journal = {Communications in Mathematical Physics},
	doi = {10.1007/BF02506410}
}

@inbook{Suzuki.1993,
	title = {Higher-Order Decomposition Theory of Exponential Operators and Its Applications to QMC and Nonlinear Dynamics},
	ISBN = {9783642784484},
	doi = {10.1007/978-3-642-78448-4_7},
	booktitle = {Computer Simulation Studies in Condensed-Matter Physics VI},
	publisher = {Springer Berlin Heidelberg},
	author = {Suzuki, Masuo and Umeno,  Ken},
	year = {1993},
	pages = {74–86}
}

@article{Suzuki.1991,
	title = {General theory of fractal path integrals with applications to many-body theories and statistical physics},
	volume = {32},
	doi = {10.1063/1.529425},
	number = {2},
	journal = {Journal of Mathematical Physics},
	publisher = {AIP Publishing},
	author = {Suzuki,  Masuo},
	year = {1991},
	pages = {400–407}
}

@article{Suzuki.1976,
	title = {Generalized Trotter’s formula and systematic approximants of exponential operators and inner derivations with applications to many-body problems},
	volume = {51},
	ISSN = {1432-0916},
	doi = {10.1007/bf01609348},
	number = {2},
	journal = {Communications in Mathematical Physics},
	publisher = {Springer Science and Business Media LLC},
	author = {Suzuki,  Masuo},
	year = {1976},
	pages = {183–190}
}

@article{Tamura.2000,
	title = {A remark on operator-norm convergence of Trotter-Kato product formula},
	volume = {37},
	ISSN = {1420-8989},
	doi = {10.1007/bf01194484},
	number = {3},
	journal = {Integral Equations and Operator Theory},
	publisher = {Springer Science and Business Media LLC},
	author = {Tamura,  Hiroshi},
	year = {2000}
}

@article{Touzard.2018,
	title={Coherent Oscillations inside a Quantum Manifold Stabilized by Dissipation},
	volume={8},
	ISSN={2160-3308},
	doi = {10.1103/physrevx.8.021005},
	number={2},
	journal={Physical Review X},
	publisher={American Physical Society (APS)},
	author={Touzard, S. and Grimm, A. and Leghtas, Z. and Mundhada, S. O. and Reinhold, P. and Axline, C. and Reagor, M. and Chou, K. and Blumoff, J. and Sliwa, K. M. and Shankar, S. and Frunzio, L. and Schoelkopf, R. J. and Mirrahimi, M. and Devoret, M. H.},
	year={2018}
}

@article{Trotter.1959,
	author = {Trotter, Hale F.},
	year = {1959},
	title = {On the Product of Semi-Groups of Operators},
	pages = {545},
	volume = {10},
	number = {4},
	issn = {00029939},
	journal = {Proceedings of the American Mathematical Society},
	doi = {c49mx4}
}

@misc{Vanluijk.2024,
	title={Energy-limited quantum dynamics}, 
	author={Lauritz van Luijk},
	year={2024},
	eprint={2405.10259},
	archivePrefix={arXiv},
	primaryClass={quant-ph},
	doi = {10.48550/arXiv.2405.10259}, 
}

@article{Vuillermot.2009,
	title = {A general Trotter–Kato formula for a class of evolution operators},
	volume = {257},
	ISSN = {0022-1236},
	doi = {10.1016/j.jfa.2009.06.026},
	number = {7},
	journal = {Journal of Functional Analysis},
	publisher = {Elsevier BV},
	author = {Vuillermot,  Pierre-A. and Wreszinski,  Walter F. and Zagrebnov,  Valentin A.},
	year = {2009},
}

@inbook{Zagrebnov.2022sumrate,
	title = {Operator-Norm Trotter Product Formula on Banach Spaces},
	ISBN = {9783031567209},
	ISSN = {2296-4878},
	DOI = {10.1007/978-3-031-56720-9_10},
	booktitle = {Operator Theory: Advances and Applications},
	publisher = {Springer Nature Switzerland},
	author = {Zagrebnov,  Valentin A. and Neidhardt,  Hagen and Ichinose,  Takashi},
	year = {2024},
	pages = {451–484}
}

@book{Zagrebnov.2019,
	title = {Gibbs Semigroups},
	ISBN = {9783030188771},
	doi = {10.1007/978-3-030-18877-1},
	journal = {Operator Theory: Advances and Applications},
	publisher = {Springer International Publishing},
	author = {Zagrebnov,  Valentin A.},
	year = {2019}
}
\end{document}